\pdfoutput=1
\documentclass[11pt]{article}
\usepackage[margin=1.15in]{geometry}
\usepackage[utf8]{inputenc}
\usepackage[T1]{fontenc}
\usepackage{amsmath, amssymb, amsthm, mathtools}
\usepackage{graphicx}
\usepackage{bm}
\usepackage[round]{natbib}
\usepackage{tikz}
\usetikzlibrary{arrows.meta, positioning}
\usepackage{pgfplots}
\pgfplotsset{compat=1.16}
\usetikzlibrary{decorations.pathreplacing, calligraphy, calc, patterns,
                backgrounds, fit, shapes.misc}
\newcommand{\pnl}[1]{\textbf{(#1)}}

\definecolor{cPrior}{RGB}{214,219,226}   
\definecolor{cPriorE}{RGB}{130,138,148}
\definecolor{cWin}{RGB}{ 31,105,175}     
\definecolor{cWinF}{RGB}{233,242,252}
\definecolor{cPost}{RGB}{ 17, 68, 50}    
\definecolor{cPostF}{RGB}{148,192,171}
\definecolor{cAdv}{RGB}{176, 30, 40}     
\definecolor{cSel}{RGB}{  0,  0,  0}     

\tikzset{
  fignum/.style   = {font=\footnotesize, text=black},
  figlab/.style   = {font=\small, text=black},
  figtiny/.style  = {font=\scriptsize, text=black},
  axisline/.style = {-{Latex[length=2mm,width=1.4mm]}, gray!70, line width=.4pt},
  band/.style     = {line width=.6pt},
  priorband/.style= {band, draw=cPriorE, fill=cPrior},
  winband/.style  = {band, draw=cWin, fill=cWinF},
  postband/.style = {band, draw=cPost, fill=cPostF},
  tick/.style     = {line width=.5pt, gray!80},
  advdot/.style   = {circle, fill=cAdv, inner sep=0pt, minimum size=3.2pt},
  seldot/.style   = {rectangle, fill=cSel, inner sep=0pt, minimum size=3pt},
  meas/.style     = {decorate, decoration={calligraphic brace, amplitude=3.2pt,
                     raise=1.2pt}, line width=.6pt},
  measlab/.style  = {font=\scriptsize, inner sep=1.5pt},
  callout/.style  = {font=\scriptsize, align=center, inner sep=1.5pt},
  cline/.style    = {-{Latex[length=1.6mm,width=1.1mm]}, gray!75, line width=.4pt},
  trajline/.style = {cAdv, line width=.9pt},
  trajdot/.style  = {circle, fill=cAdv, draw=white, line width=.4pt,
                     inner sep=0pt, minimum size=3.6pt},
}

\pgfplotsset{
  appfig/.style = {
    width=6.4cm, height=4.6cm,
    scale only axis,
    axis lines=left,
    axis line style={gray!70, line width=.4pt},
    tick style={gray!70, line width=.4pt},
    tick label style={font=\scriptsize},
    label style={font=\small},
    title style={font=\small, yshift=-1pt},
    legend style={font=\scriptsize, draw=gray!50, fill=white,
                  inner sep=2pt, row sep=-1pt},
    every axis plot/.append style={line width=.9pt},
    clip mode=individual,
  }
}
\usepackage[colorlinks=true, linkcolor=blue, citecolor=blue, urlcolor=blue]{hyperref}
\hypersetup{pdftitle={Consistent Model Chasing Is Minimax Optimal: The Exact Value of Scalar Adversarial Adaptive Control under Large Parametric Uncertainty},
            pdfauthor={Dimitar Ho}}
\usepackage{enumitem}
\usepackage{booktabs}

\newtheorem{theorem}{Theorem}
\newtheorem{proposition}{Proposition}
\newtheorem{lemma}{Lemma}
\newtheorem{corollary}{Corollary}
\theoremstyle{definition}
\newtheorem{definition}{Definition}

\theoremstyle{remark}
\newtheorem{remark}{Remark}

\newcommand{\R}{\mathbb{R}}
\newcommand{\N}{\mathbb{N}}
\newcommand{\midpt}{\operatorname{mid}}
\newcommand{\sgn}{\operatorname{sgn}}
\newcommand{\dm}{\operatorname{diam}}
\newcommand{\dH}{d_{\mathrm{H}}}
\newcommand{\norm}[1]{\left\lVert #1 \right\rVert}
\newcommand{\abs}[1]{\left\lvert #1 \right\rvert}
\newcommand{\Sel}{\texttt{SEL}}
\newcommand{\Api}{\mathcal{A}_{\pi}(\Sel)}
\newcommand{\Kclass}{\mathcal{K}}
\newcommand{\ecap}{G}   

\title{Consistent Model Chasing Is Minimax Optimal:\\
The Exact Value of Scalar Adversarial Adaptive Control\\ under Large Parametric Uncertainty}

\author{Dimitar Ho\thanks{\texttt{dho@caltech.edu}, \texttt{dimitarho@gmail.com}. This work builds on the author's doctoral research at the California Institute of Technology.}}
\date{}

\begin{document}
\maketitle

\begin{abstract}
We solve exactly a fundamental problem of adaptive control against adversarial disturbances: regulate the scalar system $x_{t+1} = a\,x_t + u_t + w_t$, $x_0=0$, $\norm{w}_\infty \le 1$, where the constant pole $a \in [-\Delta, \Delta]$ is unknown in sign and magnitude and $\Delta$ is arbitrarily large. Elementary as the system looks, the least worst-case peak $\norm{x}_\infty$ that a causal controller can guarantee against an adversarial pair $(a, w)$ (the \emph{value} of this game) has, to our knowledge, never been determined for any adaptive control problem with parametric uncertainty of arbitrary size under this criterion; existing theory supplies stability certificates, gain bounds, and regret rates, not the value. That value is
\[
\gamma^\star(\Delta) \;=\; 1 + \Delta\qquad\text{for every }\Delta>0.
\]
The summand $1$ is the irreducible price of the disturbance, and $\Delta$ the exact price of a single, unavoidable identification spike. The optimal policy is certainty-equivalent deadbeat control at the midpoint of the set-membership consistent interval, an instance of the \emph{robust oracle} $\times$ \emph{consistent model chasing} architecture \citep{ho2021online, ho2023thesis}. The architecture is forced, not merely sufficient: writing $\theta_t := -u_t/x_t$ exhibits every causal controller as an oracle--selector composition, and optimality pins the selector to the midpoint at the critical histories. The standard tools, classical and modern, each fail quantifiably at this problem: probing is punished before it pays, commitment at least doubles the optimal excess once adaptation is necessary and is fatal at sub-disturbance excitation, optimism degenerates to tie-breaking or pays asymptotically at least twice the optimum, and regret certificates are blind to the worst-case peak in both directions. The optimal controller contains no exploration mechanism of any kind --- its learning is purely passive. These results give the first exact optimality certificate for consistent model chasing as a design principle for adversarial adaptive control.
\end{abstract}

\section{Introduction}\label{sec:intro}

The core promise of adaptive control is one-shot autonomy: deploy a single control algorithm on a system whose dynamics are largely unknown, and guarantee acceptable closed-loop behavior from the first time step, while the algorithm learns whatever it needs online. When the environment is adversarial (disturbances chosen by nature to do maximal harm, with the unknown dynamics conspiring with them), the appropriate notion of guarantee is a \emph{worst-case} one, and the appropriate formulation is a dynamic game between the controller and an adversary that controls both the disturbance and the residual model uncertainty.

This paper studies what is arguably the minimal nontrivial instance of that game, and solves it exactly. The instance is chosen for its simplicity, not despite it: everything about the problem fits in one line, yet its exact minimax value has remained open (posed as a min--max program, for parameter polytopes, in \citealp{ho2018passive}, which solved its static-feedback restriction and left the value uncomputed), and it is the smallest arena in which the defining difficulty of adversarial adaptation appears in pure form. Consider the scalar linear system
\begin{equation}\label{eq:sys}
    x_{t+1} \;=\; a\,x_t + u_t + w_t, \qquad x_0 = 0,
\end{equation}
with full state observation, adversarial disturbance $\norm{w}_\infty \le 1$, and a constant but \emph{unknown} pole $a$, known only to satisfy $a \in [-\Delta, \Delta]$. The half-width $\Delta$ may be arbitrarily large: the open-loop system can be strongly unstable, with unknown magnitude \emph{and unknown sign}. The controller is any causal law $u_t = K_t(x_t, \dots, x_0)$, and the objective is the worst-case peak of the state:
\begin{equation}\label{eq:minimax}
    \gamma^\star(\Delta) \;=\;
    \inf_{K \in \Kclass}\;
    \sup_{a \in [-\Delta,\Delta]}\;
    \sup_{\norm{w}_\infty \le 1}\;
    \norm{x(K,a,w)}_\infty .
\end{equation}
This is a one-shot, worst-case, large-uncertainty problem: for $\Delta \ge 1$ no static linear controller achieves a finite value, so adaptation (closed-loop learning) is not an optimization but a prerequisite for boundedness. At the same time the problem couples learning and control as tightly as possible: information about $a$ is only generated when the state is away from the origin, i.e., precisely when the system is exposed.

\paragraph{Main result.} We prove (Theorem~\ref{thm:main}) that
\[
    \boxed{\;\gamma^\star(\Delta) = 1 + \Delta\quad\text{for every }\Delta > 0,\;}
\]
with both the infimum and the suprema attained. To our knowledge this is the first \emph{exact} minimax value (not an order-of-growth, not a regret rate) for adaptive control with interval parametric uncertainty and adversarial disturbances under the worst-case peak criterion. The two summands are priced exactly ($1$ for the disturbance, $\Delta$ for a single, unavoidable \emph{identification spike}), an anatomy made precise after Theorem~\ref{thm:main}.

\paragraph{The optimal policy is an oracle--selector architecture.} The second contribution is structural. The optimal controller (Definition~\ref{def:controller}) maintains the interval $A_t \subseteq [-\Delta,\Delta]$ of parameter values consistent with the observed data, and plays certainty-equivalent deadbeat control at the \emph{midpoint} of that interval:
\[
    u_t = -\midpt(A_t)\, x_t .
\]
This is precisely an instantiation of the modular \emph{online robust control} framework introduced in \citet*{ho2021online} and developed in the author's thesis as the ``PixSel'' (``$\pi \times \Sel$'') framework \citep[Ch.~6]{ho2023thesis}: decompose the large-uncertainty control problem into (i) a \emph{robust oracle} $\pi$ (here, the deadbeat law $\pi[\theta](x) = -\theta x$, the nominal design that would be optimal if the parameter were known) and (ii) a \emph{consistent model chaser} $\Sel$ --- here, the midpoint (in dimension one, also the Steiner-point) selection on the set-membership consistent set, which is a $1$-competitive chasing algorithm (competitiveness in the precise sense of \citealp{ho2021online}, the constant proved in Appendix~\ref{app:interface}). The framework's general theorems guarantee that any such composition $\Api$ enjoys worst-case safety and finite-mistake bounds \emph{for arbitrarily large uncertainty}; the general bounds are, necessarily, conservative.

What Theorem~\ref{thm:main} shows is much stronger and, we believe, new in kind: for the fundamental problem \eqref{eq:minimax}, the oracle--selector composition is not merely sufficient --- \emph{it is exactly minimax optimal, and its two components are individually forced}. Indeed the decomposition is a \emph{normal form} rather than a modeling choice: writing $\theta_t := -u_t/x_t$ exhibits every causal law as an oracle--selector composition with some, possibly inconsistent, selector (Definition~\ref{def:implicit}). Measured against that yardstick, Definition~\ref{def:controller} is forced in three respects: passivity before information exists ($u_0 = 0$, and then $u_1 = 0$ after the kick; Remark~\ref{rem:u0}), the exposure cap as a necessary condition on every optimal controller rather than a property of our particular law (Proposition~\ref{prop:implicit}), and the midpoint as the \emph{unique} optimal selection at the critical step, every optimism-style bias toward an extreme being strictly and quantifiably worse (Section~\ref{sec:exploration}). Consistent model chasing is thus not one design choice among many for this problem class: it is the solution.

We are deliberate about scope. The architecture is not scalar: the framework and its worst-case guarantees are formulated for general nonlinear systems and compact model classes \citep{ho2021online}. What is scalar is the \emph{completeness} of the certificate (exact value, forced architecture, priced impossibilities), and whether that exactness extends to higher dimensions is open: the game there is genuinely richer (Section~\ref{sec:outlook}).

\paragraph{Why the standard tools fail here.} The optimal policy performs no exploration whatsoever: no probing signal, no persistency of excitation, no optimism, no explicit identification phase. All learning is \emph{passive}: the controller merely refuses to contradict the data (consistency) and moves its model hypothesis as little as the data allow (competitiveness). Section~\ref{sec:exploration} proves that this is forced. The standard tools, classical and modern, are each strictly suboptimal for the problem \eqref{eq:minimax}:
\begin{itemize}[nosep]
    \item \emph{Probing is punished before it pays} (Proposition~\ref{prop:probing}): a probe is amplified by the full uncertainty $\Delta$ one step before its information can be used, and the penalty is attained exactly.
    \item \emph{Commitment can be gamed} (Proposition~\ref{prop:etc}): every scheme that irrevocably commits to a frozen gain in finite time (explore-then-commit in particular) pays, for $\Delta \ge 1$, at least double the optimal excess, and diverges whenever its excitation stays at or below the disturbance level: a cancellation adversary erases that excitation, starving the identification phase outright.
    \item \emph{Optimism cannot select} (Propositions~\ref{prop:endpoint}--\ref{prop:ofudichotomy}): under the peak objective the clairvoyant value function is identical for every consistent model, so optimism degenerates to tie-breaking; and every instantiation that actually determines the selection pays asymptotically at least twice (endpoint rules four times) the optimum. The midpoint is the unique optimal selection, and it is a set-relative, minimax object that no per-model criterion selects.
\end{itemize}
Their common root is an information-timing asymmetry intrinsic to adversarial adaptive control, displayed and exploited in Section~\ref{sec:exploration}. Section~\ref{sec:regret} closes the case from the metric side: regret, the currency in which exploration is usually justified, is blind to the worst-case peak in \emph{both} directions (Proposition~\ref{prop:regret}).

\paragraph{Closing a loop.} This work closes, at its foundational instance, a loop that began with a line of research on one-shot control design with worst-case guarantees under arbitrarily large model uncertainty: passive--aggressive learning and control \citep{ho2018passive}, model-free robust stabilization without prior bounds \citep{ho2019robust}, scalable adaptive system-level control \citep{ho2019scalable}, and the oracle/consistent-model-chasing decomposition \citep*{ho2021online}, subsequently extended to distributed systems \citep*{yu2023online} and applied to linear time-varying systems (where the consistent sets are no longer nested and the selector becomes general convex body chasing) \citep*{yu2023onlineLTV}, and to power systems \citep*{yeh2022online}. The program is consolidated in the author's thesis \citep{ho2023thesis}, which contains the general theory in full and includes results not published elsewhere. That framework was built as a \emph{sufficient} design methodology: robust control supplies the oracle, competitive online learning (convex body chasing, \citealp{friedman1993convex, bubeck2020chasing, argue2019nearly, sellke2020chasing}) supplies the selector, and the composition inherits worst-case guarantees. The present paper supplies the missing converse for the scalar problem: the architecture is optimal, exactly. We view this as evidence that consistent model chasing (a consistent selector that moves no more than the data compel, rather than exploration) is the correct first-order principle for adversarial adaptive control at large uncertainty, in the same way that certainty equivalence, itself exactly optimal only in the LQG setting, is the correct first-order principle of the stochastic theory.

\subsection{Related work}\label{sec:related}

\paragraph{Minimax and dual adaptive control.} The idea of treating adaptive control as a dynamic game goes back to dual control \citep{feldbaum1960dual, filatov2004dual} and was given a modern game-theoretic form by \citet{didinsky1994minimax} via cost-to-come functions for soft-constrained quadratic criteria. A contemporary line, minimax adaptive control in the sense of Rantzer and Kjellqvist, certifies closed-loop $\ell_2$-gains --- first for \emph{finite} sets of linear models \citep{rantzer2021minimax, kjellqvist2022noisy}, with an exact solution for a state matrix of unknown sign \citep{rantzer2020sign}, and subsequently for norm-bounded continua of parameters in exactly solvable fully actuated and single-input cases \citep{rantzer2025actuated, rantzer2026single} --- together with finite-dimensional information states for minimax dual control \citep{kjellqvist2024dual}, its output-feedback extension \citep{kjellqvist2024output}, and regret analyses thereof \citep{renganathan2023regret}. We work in the same game spirit, but in a different \emph{currency}. An $\ell_2$-gain certificate bounds accumulated energy, so a single large excursion may be amortized against many small ones; ours is the induced $\ell_\infty$ peak (the amortization-free currency of safety), and the deliverable is not a certified controller but the exact value of the game, the exact optimal policy attaining it, its forced oracle--selector architecture, and the impossibility theorems of Section~\ref{sec:exploration}. The known-dynamics theory of this criterion is $\ell_1$ optimal control \citep{vidyasagar1986optimal, dahleh1987l1, dahleh1995control}: the floor $1$ in Theorem~\ref{thm:main} is the known-parameter $\ell_1$-optimal value, and $\Delta$ is what not knowing the parameter adds to it. To our knowledge no prior work gives an exact minimax \emph{peak} value for parametric uncertainty of arbitrary size. Competitive-ratio formulations of online control with \emph{known} dynamics appear in \citet{goel2021competitive}: closed-loop cost is compared multiplicatively against the hindsight-optimal controller. Our use of the word ``competitive'' is unrelated: it is that of the chasing literature \citep{friedman1993convex, bubeck2020chasing, sellke2020chasing}, and it enters only through the model \emph{selector}, whose movement is bounded by the Hausdorff variation $\dH$ of the consistent sets, a criterion internal to the model-chasing subproblem, defined without reference to any control cost or clairvoyant benchmark. The controller-level benchmark in this paper is instead the exact minimax value; indeed, Proposition~\ref{prop:regret} shows that benchmark-relative certificates cannot carry the worst-case guarantee sought here.

\paragraph{Set-membership identification.} The information structure of our problem (each transition confines the parameter to a data window, and uncertainty sets are nested intersections) is precisely that of bounded-noise set-membership identification \citep{fogel1982value, MILANESE1991997, tse1993optimal}. That literature quantified identification under bounded noise, including its fundamental hardness \citep{dahleh1993sample}, but as an estimation problem in its own right, decoupled from any control objective. Set-membership estimation has since been placed inside worst-case control loops (robust adaptive MPC \citep{lorenzen2019robust, kohler2021robust, parsi2020robust} being the most developed instance, optimizing a worst-case cost over the feasible parameter set at every step). Those guarantees, however, are premised on the uncertainty set being robustly stabilizable by a \emph{single fixed} controller, which is the \emph{small}-uncertainty regime: in \eqref{eq:minimax} a fixed gain $u = -kx$ achieves $\abs{a - k} < 1$ for every $a \in [-\Delta, \Delta]$ precisely when $\Delta < 1$, the regime in which, as the time-varying comparison below shows, even an adversarially varying parameter is benign and adaptation buys nothing. Multi-model and switching schemes \citep{anderson2001multiple, hespanha2003overcoming} do reach beyond it, at the cost of any worst-case optimality certificate. The gap is structural: even restricted to $\Delta < 1$, none of this work determines the exact minimax peak; it certifies stability, constraint satisfaction, or a receding-horizon worst-case cost, and leaves the value of the game open. Theorem~\ref{thm:main} settles that value, for every compact interval and every $\Delta$. Our window lemma (Lemma~\ref{lem:window}) can be read as a closed-loop accounting identity: realized state magnitude and surviving parameter uncertainty are exchanged one-for-one against the same budget. Sokolov's $\ell_1$-adaptive line \citep{sokolov1985adaptive, sokolov2001adaptive} put set-membership estimators in the loop with matching norm structure, obtaining asymptotic suboptimality bounds; our contribution is the exact transient value. An instructive contrast, in the spirit of the feedback-capability results of \citet{xie2000how}: for a \emph{time-varying} pole $a_t \in [-\Delta, \Delta]$ the value of \eqref{eq:minimax} is finite if and only if $\Delta < 1$,\footnote{The adversary aligns $\sgn(a_t x_t)$ with $u_t$ at every step, forcing $\abs{x_{t+1}} \ge \Delta\abs{x_t} + 1$, which is unbounded for $\Delta \ge 1$; conversely $u \equiv 0$ keeps the state bounded for $\Delta < 1$.} whereas for a constant pole it is finite --- and equals $1 + \Delta$ --- for every $\Delta$. The entire benefit of large-uncertainty adaptation is thus attributable to the parameter being \emph{identifiable}, and our theorem prices that benefit exactly.

\paragraph{Online learning for control.} A large recent literature studies learning to control linear systems with regret guarantees, predominantly for stochastic or small-uncertainty settings: LQR regret \citep{abbasi2011regret, dean2018regret, cohen2018online, cohen2019learning, mania2019certainty}, online control with adversarial disturbances for \emph{known} or stabilized systems \citep{agarwal2019online, hazan2019nonstochastic}, and black-box control of unknown systems \citep{chen2021blackbox}, the latter proving an exponential-in-dimension lower bound that is the nearest regret-metric analogue of the exponential-in-dimension phenomena we expect beyond the scalar case (Section~\ref{sec:outlook}). For stochastic LQR this literature has itself concluded that optimism is unnecessary: certainty equivalence with naive exploration is regret-optimal \citep{simchowitz2020naive}. Section~\ref{sec:exploration} proves the adversarial, exact counterpart: for the problem \eqref{eq:minimax}, every one of the standard tools is not merely unnecessary but strictly suboptimal, and the penalty is larger than the amortization intuition suggests, ranging from a constant factor of two to four up to unboundedness. We bound each tool's price from below --- exactly, in the case of probing. Our Proposition~\ref{prop:regret} explains why regret is structurally unable to certify the guarantee sought here, extending the mistake-versus-regret separation observed in \citet{ho2021online}. Classical robust adaptive control \citep{ioannou1996robust, AstromAdaptive} and multi-model adaptive control \citep{anderson2001multiple, hespanha2003overcoming} pursue stability under large uncertainty but do not provide worst-case optimality certificates.

\paragraph{Competitive selection and convex body chasing.} The selector of the optimal policy is a competitive chasing algorithm on nested consistent sets, connecting adaptive control to the competitive-analysis literature: convex body chasing \citep{friedman1993convex}, whose general case was settled with linear \citep{argue2020chasing} and then optimal \citep{sellke2020chasing} competitive ratios, and its \emph{nested} case (the only case this paper requires), for which \citet{argue2019nearly} gave a nearly-linear bound and \citet{bubeck2020chasing} the near-optimal Steiner-point algorithm. The reduction of adaptive-control model selection to nested convex body chasing was introduced in \citet{ho2021online}; here, the one-dimensional Steiner point (the interval midpoint) turns out to be a component of the \emph{exactly optimal} policy.

\subsection{Notation}
$\N = \{0,1,\dots\}$. For a sequence $x = (x_0, x_1, \dots)$, $\norm{x}_\infty := \sup_t \abs{x_t}$. For a bounded interval $I \subseteq \R$, $\abs{I}$ denotes its length and $\midpt(I)$ its midpoint. $\dH$ denotes the Hausdorff distance. All proofs not given in the main text appear in the appendix.

\section{Problem formulation}\label{sec:problem}

Consider the system \eqref{eq:sys} with state $x_t \in \R$ measured exactly, input $u_t \in \R$, adversarial disturbance $\norm{w}_\infty \le 1$, and unknown constant parameter $a \in [-\Delta, \Delta]$ with $\Delta > 0$ known. Admissible controllers are causal, deterministic feedback laws
\begin{equation}\label{eq:causal}
    u_t = K_t(x_t, x_{t-1}, \dots, x_0),
\end{equation}
collected in the class $\Kclass$; the value of the game is $\gamma^\star(\Delta)$ defined in \eqref{eq:minimax}. Two normalizations are without loss: rescaling $(x, u, w)$ by the disturbance bound reduces any bound $\norm{w}_\infty \le W$ to $1$, scaling the value by $W$; and the causal substitution $\tilde u_t = u_t + c\,x_t$ recenters any prior interval at the origin, so the symmetry of $[-\Delta, \Delta]$ is a convention rather than an assumption (Remark~\ref{rem:extensions}(c) records the resulting value for an asymmetric prior). Three remarks frame the problem.

\begin{remark}[Adaptation is necessary]\label{rem:adapt}
A static law $u_t = -kx_t$ yields $x_{t+1} = (a-k)x_t + w_t$, stable only if $\abs{a-k}<1$; a single gain works for all $a \in [-\Delta,\Delta]$ only when $\Delta < 1$. For $\Delta \ge 1$, therefore, no static linear law achieves a finite worst-case peak; the controllers of this paper escape the obstruction by genuinely using the history to infer $a$.
\end{remark}

\begin{remark}[Learning is implicit in the loop]
Each transition localizes the parameter: whenever $x_t \neq 0$, the observation $(x_t, x_{t+1} - u_t)$ confines $a$ to a \emph{data window} of width $2/\abs{x_t}$,
\begin{equation}\label{eq:window}
    a \;\in\; W_t := \left\{\alpha : \abs{x_{t+1} - \alpha x_t - u_t} \le 1 \right\}.
\end{equation}
Information about $a$ is therefore generated exactly when, and in proportion to how far, the state is away from the origin --- the crux of the exploration--regulation coupling.
\end{remark}

\begin{remark}[Peak objective]
The $\ell_\infty$-peak criterion is the natural carrier of worst-case safety guarantees (state constraints that must never be violated); it is also the least forgiving to learning transients, since a single bad step is never amortized. This is precisely the regime where, as we show in Sections~\ref{sec:exploration} and~\ref{sec:regret}, intuitions imported from average-cost and regret formulations fail.
\end{remark}

\section{Main result: the exact value and the optimal architecture}\label{sec:main}

\begin{theorem}[Exact minimax value]\label{thm:main}
For every $\Delta > 0$,
\begin{equation}\label{eq:value}
    \gamma^\star(\Delta) \;=\; 1 + \Delta,
\end{equation}
and the value is attained by the controller of Definition~\ref{def:controller}.
\end{theorem}

The two summands have a clean interpretation: $1$ is the irreducible effect of the disturbance (unavoidable even when $a$ is known), and $\Delta$ is the exact price of a single, unavoidable \emph{identification spike}, so named because realizing it is what identifies $a$ exactly (Remark~\ref{rem:anatomy}), not because anyone seeks it: the disturbance can move the state from $0$ to $\pm1$ without leaking any information about $a$, after which one step of exposure of magnitude $\Delta$ cannot be prevented. The content of the theorem is that an optimal controller pays this price \emph{at most once}.

\subsection{The optimal controller}

\begin{definition}[Set-membership certainty-equivalent deadbeat controller]\label{def:controller}
Maintain the set $A_t \subseteq \R$ of parameter values consistent with the data observed up to time $t$:
\begin{equation}\label{eq:setmem}
    A_0 = A_1 = \left[-\Delta, \Delta\right], \qquad
    A_{t+1} \;=\;
    \begin{cases}
        A_t \cap W_t, & x_t \neq 0,\\[2pt]
        A_t, & x_t = 0,
    \end{cases}
\end{equation}
with $W_t$ the data window \eqref{eq:window}. Each $A_t$ is a closed interval containing the true parameter, computable by the controller at time $t$ (it depends only on $x_0, \dots, x_t$ and $u_0, \dots, u_{t-1}$). Write $\hat a_t := \midpt(A_t)$ and $L_t := \abs{A_t}$ (so $L_0 = L_1 = 2\Delta$). The control law is
\begin{equation}\label{eq:law}
    u_t \;=\; -\,\hat a_t\, x_t .
\end{equation}
Note $u_0 = 0$ automatically, since $x_0 = 0$.
\end{definition}

The law \eqref{eq:law} is certainty-equivalent deadbeat control at the midpoint of the membership set: it cancels the best available estimate of $a x_t$ and places the nominal successor state at the origin. With the estimation error $e_t := a - \hat a_t$ (so $\abs{e_t} \le L_t/2$), the closed loop reads
\begin{equation}\label{eq:closedloop}
    x_{t+1} \;=\; e_t\,x_t + w_t .
\end{equation}

\subsection{The controller is an oracle--selector composition}\label{sec:architecture}

Definition~\ref{def:controller} is an instance of the meta-algorithm $\Api$ of \citet{ho2021online}: at each step, present the data to a selector $\Sel$ to obtain a posited parameter $\theta_t$, then apply the policy $\pi[\theta_t]$ returned by a nominal-design oracle $\pi$:
\begin{center}
\begin{tikzpicture}[node distance=7mm and 12mm, every node/.style={font=\small}]
    \node[draw, rounded corners, minimum width=2.9cm, minimum height=8.5mm] (sel) {$\Sel$: $\;\theta_t = \midpt(A_t)$};
    \node[draw, rounded corners, minimum width=2.9cm, minimum height=8.5mm, right=of sel] (pi) {$\pi$: $\;u_t = -\theta_t\, x_t$};
    \node[draw, rounded corners, minimum width=2.9cm, minimum height=8.5mm, right=of pi] (sys) {$x_{t+1} = a x_t + u_t + w_t$};
    \draw[-{Latex}] (sel) -- node[above]{$\theta_t$} (pi);
    \draw[-{Latex}] (pi) -- node[above]{$u_t$} (sys);
    \draw[-{Latex}] (sys.south) -- ++(0,-5.5mm) -| node[below, pos=0.25]{data $(x_{s+1}, x_s, u_s)_{s < t}$ $\;\to\;$ consistent set $A_t$} (sel.south);
\end{tikzpicture}
\end{center}
The two components of Definition~\ref{def:controller} provably satisfy the framework's two interface conditions, stated for general nonlinear systems and compact parameter spaces in \citet{ho2021online}. Informally: the deadbeat law $\pi[\theta](x) = -\theta x$ is a \emph{$\rho$-robust oracle} for every $\rho \in (0,1)$ --- applied with any parameter sequence within $\rho$ of the true $a$, it confines the state to the geometric envelope $\abs{x_{t+1}} \le \rho\abs{x_t} + 1$, with a logarithmic mistake function for every threshold above the envelope's limit (a \emph{mistake}: a step at which $\abs{x_t}$ exceeds the threshold; Proposition~\ref{prop:oracle}); and the midpoint is a \emph{$1$-competitive consistent model chaser} on the nested consistent sets. It always selects a parameter consistent with all data seen so far, and its total movement over any time window is bounded by the Hausdorff distance between the consistent sets at the window's ends, and by $\Delta$ in total over the infinite horizon (Proposition~\ref{prop:selector}(ii)--(iii)). Appendix~\ref{app:interface} states the interface conditions formally (Definition~\ref{def:interface}) and proves both claims (Propositions~\ref{prop:oracle} and~\ref{prop:selector}).

The general theory of \citet{ho2021online} guarantees, for any such composition, worst-case boundedness and finite mistake bounds that are polynomial in $\Delta$ --- the framework's own scalar design example instantiates this pair (deadbeat oracle, Steiner selector) and derives an $O(\Delta^2)$ mistake guarantee. Theorem~\ref{thm:main} upgrades the qualitative guarantee for the problem \eqref{eq:minimax}: this composition attains the \emph{exact} minimax peak. Moreover, the decomposition is not a restriction of the design space at all:

\begin{definition}[Implicit selector]\label{def:implicit}
At any history with $x_t \neq 0$, the action of an arbitrary causal law $K \in \Kclass$ defines its \emph{implicit selector}
\[
    \theta_t \;:=\; -\,\frac{u_t}{x_t}, \qquad\text{so that}\qquad u_t \;=\; -\theta_t\,x_t \quad\text{identically.}
\]
No consistency is asserted: $\theta_t$ need not lie in $A_t$, nor even in the prior. (The recursion \eqref{eq:setmem} reads off observed data alone, so the consistent set $A_t$ --- an intersection of intervals, hence an interval --- is defined along the trajectory of \emph{any} causal law, not only of the law \eqref{eq:law}.)
\end{definition}

The identification is a tautology, and that is precisely its use: it converts a statement about a restricted \emph{class} of controllers into a statement about a \emph{quantity} that every controller possesses. Every causal law factors through the deadbeat oracle $\pi$ and some selector; the composition $\Api$ is a normal form for this problem, and the only question a design can answer is \emph{which} selector. Definition~\ref{def:controller} answers it with a selector that is consistent and centered, and both properties are forced: Proposition~\ref{prop:implicit} shows that the worst-case peak of any causal $K$ is at least $1 + s_t\sup_{\alpha \in A_t}\abs{\alpha - \theta_t}$ at every reachable history (writing $s_t := \abs{x_t}$ and $g_t := L_t s_t$ for the \emph{exposure}, the diameter of $\{\alpha x_t : \alpha \in A_t\}$ and thus how far the surviving uncertainty can still throw the state in one step, both formally introduced in Section~\ref{sec:ub}), so optimality caps every controller's exposure and, where that exposure is maximal, pins $\theta_t$ to the midpoint. At $x_t = 0$ the identification is vacuous; there Proposition~\ref{prop:probing} supplies the complementary forcing $u_0 = 0$ at the start. Concretely:

\begin{remark}[Passivity is necessary: $u_0 = 0$, and then $u_1 = 0$]\label{rem:u0}
Every optimal controller must play $u_0 = 0$: by Proposition~\ref{prop:probing} below, $\abs{u_0} = \eta > 0$ implies worst-case peak at least $1 + (1+\eta)\Delta > \gamma^\star(\Delta)$. The forcing extends one further step. After the kick $w_0 = \pm 1$ the state is $x_1 = \pm 1$ with the prior intact ($A_1 = A_0$, the regressor $x_0$ having leaked nothing), so the exposure $g_1 = 2\Delta$ is maximal and Proposition~\ref{prop:implicit}(iii) applies: every optimal causal controller plays $u_1 = 0$ \emph{exactly}, whatever its internal parameterization, and a deviation of size $d$ costs exactly $d$ --- worst-case peak $1 + \Delta + d$ (the lower bound is Proposition~\ref{prop:implicit}(iii); the matching upper bound is Remark~\ref{rem:u1exact}).
\end{remark}

\begin{remark}[The midpoint is the unique optimal selection]\label{rem:midunique}
At the critical step the midpoint is forced among \emph{all} causal controllers, not merely among certainty-equivalent deadbeat laws: by Definition~\ref{def:implicit} every law has an implicit selector $\theta_t = -u_t/x_t$, and Proposition~\ref{prop:implicit}(iii) leaves it no freedom at any history of maximal exposure. Proposition~\ref{prop:endpoint} prices the deviation as a function of the selector's relative position in the consistent interval, uniquely minimized at the midpoint. The forcing is local: it binds precisely on the maximal-exposure set $\{g_t = 2\Delta\}$ (which the optimal law of Definition~\ref{def:controller} visits only at the kick-from-rest histories), and elsewhere leaves the explicit tolerance displayed after Proposition~\ref{prop:implicit} in Section~\ref{sec:cap}, an outer bound that becomes positive strictly one step after the peak. Away from the midpoint the loss also compounds: an endpoint selector suffers at least $4\Delta - 1$ (Proposition~\ref{prop:echo}).
\end{remark}

\section{Lower bound: no causal controller beats \texorpdfstring{$1+\Delta$}{1+Delta}}\label{sec:lb}

\begin{proposition}[Lower bound]\label{prop:lb}
For every causal controller $K \in \Kclass$ there exist an admissible parameter $a \in [-\Delta,\Delta]$ and disturbance $\norm{w}_\infty \le 1$ such that the closed-loop trajectory satisfies $\norm{x}_\infty \ge 1 + \Delta$. Consequently $\gamma^\star(\Delta) \ge 1 + \Delta$.
\end{proposition}

\begin{proof}
Fix $K = (K_0, K_1, \dots)$ and construct $(a,w)$ explicitly.

\emph{Step 1 (free, information-less kick to magnitude $\ge 1$).} The number $u_0 = K_0(0)$ is determined by $K$ alone. Choose
\[
    w_0 := \begin{cases} +1, & u_0 \ge 0,\\ -1, & u_0 < 0,\end{cases}
    \qquad\text{so that}\qquad
    x_1 = u_0 + w_0,\quad \abs{x_1} = \abs{u_0} + 1 \ge 1 .
\]
Because $x_0 = 0$, the transition $x_1 = a\cdot 0 + u_0 + w_0$ carries \emph{no information about $a$}: every $a \in [-\Delta, \Delta]$ is consistent with the data $(x_0, x_1)$, whatever $w_0$ did. Note also that $x_1$ does not depend on $a$, so the controller's next move $u_1 = K_1(x_1, x_0)$ is a fixed number once $w_0$ is fixed.

\emph{Step 2 (aligned extreme parameter).} Choose the extreme parameter whose contribution reinforces $u_1$:
\[
    a^\star := \Delta\,\sigma, \qquad
    \sigma := \begin{cases} \sgn(u_1)\,\sgn(x_1), & u_1 \neq 0,\\ \sgn(x_1), & u_1 = 0,\end{cases}
\]
so that $a^\star x_1$ and $u_1$ have the same sign (or $u_1 = 0$). Then $\abs{a^\star x_1 + u_1} = \Delta\abs{x_1} + \abs{u_1} \ge \Delta$. Finally take $w_1 := \sgn(a^\star x_1 + u_1) \in \{\pm1\}$ and $w_t := 0$ for $t \ge 2$. Then
\[
    \abs{x_2} \;=\; \abs{a^\star x_1 + u_1 + w_1} \;=\; \Delta\abs{x_1} + \abs{u_1} + 1 \;\ge\; 1 + \Delta . \qedhere
\]
\end{proof}

\begin{remark}[Two extremes suffice; randomization does not help]\label{rem:randomization}
The argument uses only the two extreme parameter values $\pm\Delta$ and two disturbance moves, and exploits a single structural fact: the state can be moved off the origin \emph{before} any information about $a$ exists. Randomization does not help: the construction reacts only to realized past states and inputs (with the parameter itself chosen only at $t = 1$, which is legitimate since no data restricts $a$ before then), so the bound extends to randomized controllers against a causal disturbance strategy; moreover, a Jensen argument (align $w_0$ with $\mathbb{E}[u_0]$, the extreme parameter with $\mathbb{E}[u_1]$, and $w_1$ with the resulting mean drift) shows that even an \emph{oblivious} disturbance sequence, chosen with knowledge of the randomized strategy but not of its realizations, forces $\mathbb{E}\,\norm{x}_\infty \ge 1 + \Delta$.
\end{remark}

\section{Upper bound: one identification spike is all the adversary gets}\label{sec:ub}

\subsection{The window lemma, the invariant, and the upper bound}

Throughout this section the controller is the law \eqref{eq:law}. Write $s_t := \abs{x_t}$ and define the \emph{exposure}
\[
    g_t \;:=\; L_t\,s_t \;=\; \dm\{a x_t : a \in A_t\},
\]
the functional diameter of the surviving uncertainty along the current state. The exposure governs the one-step reachable set: since the true $a$ lies in $A_t$ and $\abs{e_t} \le L_t/2$, the closed loop \eqref{eq:closedloop} gives
\begin{equation}\label{eq:reach}
    \abs{x_{t+1}} \;\le\; \frac{g_t}{2} + 1
    \qquad\text{for every $t$ and every admissible $(a,w)$.}
\end{equation}
The heart of the proof is that $g_t$ can never exceed $2\Delta$: increasing the state costs information about $a$ at exactly the rate needed to keep the exposure bounded. Everything follows from one elementary statement about intersecting an interval with a width-$2$ window.

\begin{lemma}[Window lemma: the information--exposure tradeoff]\label{lem:window}
Let $r \ge 0$ and $y \in \R$, and suppose $I := [-r,r] \cap [y-1, y+1]$ is nonempty, with $\mu := \abs{I}$. Then
\begin{equation}\label{eq:tradeoff}
    \text{(i)}\quad \abs{y} + \mu \;\le\; r + 1,
    \qquad\qquad
    \text{(ii)}\quad \mu\,\abs{y} \;\le\; 2r .
\end{equation}
\end{lemma}

Part (i) says realized magnitude and surviving uncertainty are exchanged one-for-one against the same budget; part (ii) is the multiplicative form that yields the key invariant in one line.

\begin{proof}
Reflecting $(y,I) \mapsto (-y,-I)$ if necessary, assume $y \ge 0$; nonemptiness gives $y \le r+1$, and $\mu \le \min(2, 2r)$ because $I$ lies in both intervals. Three exhaustive cases.
\begin{itemize}[nosep]
    \item \emph{Window interior} ($y+1 \le r$): then also $y - 1 \ge -r$ (since $y \ge 0$ and $r \ge y+1 \ge 1$), so $\mu = 2$ and $y \le r-1$, whence $y + \mu \le r+1$ and $\mu y = 2y \le 2(r-1) < 2r$.
    \item \emph{Right-clipped} ($y - 1 \ge -r$, $y+1 > r$): then $I = [y-1, r]$, so $\mu = r+1-y$ and (i) holds \emph{with equality}. Substituting $y = r+1-\mu$,
    \[
        2r - \mu y \;=\; 2r - \mu(r+1-\mu) \;=\; r(2-\mu) + \mu(\mu - 1).
    \]
    If $\mu \ge 1$, both terms are nonnegative (recall $\mu \le 2$). If $\mu < 1$, use $r \ge \mu/2$: $r(2-\mu) - \mu(1-\mu) \ge \tfrac\mu2(2-\mu) - \mu(1-\mu) = \tfrac{\mu^2}{2} \ge 0$.
    \item \emph{Window covers the left end} ($y-1 < -r$, possible only if $r<1$ since $y \ge 0$): then $\mu \le 2r$, so $y + \mu < (1-r) + 2r = 1+r$, and $\mu y \le 2r(1-r) \le 2r$. \qedhere
\end{itemize}
\end{proof}

The lemma applies to the closed loop through a change of variables identifying the posterior consistent set with such an intersection.

\begin{corollary}[Posterior geometry]\label{cor:geom}
Fix $t$ with $x_t \neq 0$ and set $r := g_t/2$, $y := x_{t+1}$, and $\mu_{t+1} := s_t L_{t+1}$ (the posterior length in regressor units). Then $\mu_{t+1} = \abs{[-r,r] \cap [y-1, y+1]}$, so \eqref{eq:tradeoff} holds with these values; in particular
\[
    \text{(i)}\quad \abs{x_{t+1}} + s_t L_{t+1} \;\le\; \frac{g_t}{2} + 1,
    \qquad\qquad
    \text{(ii)}\quad s_t L_{t+1}\,\abs{x_{t+1}} \;\le\; g_t .
\]
\end{corollary}

\begin{proof}
Consider the affine change of variables $\varphi_t : \R \to \R$, $\varphi_t(\alpha) := (\alpha - \theta_t)x_t$, written for a general selected point $\theta_t$; for the midpoint law considered here $\theta_t = \hat a_t$ (both vary with $t$; the index is fixed throughout this proof). Since $A_t$ is an interval centered at $\hat a_t = \theta_t$, its image is symmetric: $\varphi_t(A_t) = [-r, r]$ with $r = L_t s_t/2 = g_t/2$. The update \eqref{eq:setmem} intersects $A_t$ with the data window $W_t$, an interval of width $2/\abs{x_t}$ in $\alpha$-space; since $\alpha x_t + u_t = (\alpha - \hat a_t)x_t = \varphi_t(\alpha)$ and $x_{t+1} = e_t x_t + w_t$, its image is $\varphi_t(W_t) = [x_{t+1} - 1, x_{t+1}+1]$, of width $2$. Hence the posterior, in the $\varphi_t$-variable, is
\[
    I \;=\; \varphi_t(A_{t+1}) \;=\; \varphi_t(A_t) \cap \varphi_t(W_t) \;=\; [-r,r] \cap [y-1, y+1], \qquad y := x_{t+1},
\]
using that the bijection $\varphi_t$ commutes with intersection. Because $\varphi_t$ is affine with slope $x_t \neq 0$, it scales interval lengths by $\abs{x_t}$, so $\abs{I} = s_t L_{t+1} = \mu_{t+1}$. (The scale factor is the \emph{current} state $x_t$ --- the regressor through which the observed transition measures $a$ --- not the successor $x_{t+1}$, which only positions the window.) $I$ is nonempty because the true parameter lies in both intervals. Figure~\ref{fig:errorcoord} pictures the change of variables and the exchange it exposes.
\end{proof}

\begin{figure}[tbp]
\centering
\begin{tikzpicture}[x=1cm, y=1cm]

\begin{scope}
  \node[figlab, anchor=south west] at (-0.15,3.60) {\pnl{a} model space $\alpha$};
  \draw[axisline] (-0.1,-0.35) -- (6.35,-0.35)
        node[fignum, above left=-1pt and -2pt] {$\alpha$};

  \draw[gray!50, dashed, line width=.4pt] (2.6,0.43) -- (2.6,2.47);
  \draw[gray!50, dashed, line width=.4pt] (4.1,0.43) -- (4.1,2.47);
  \draw[cAdv, dash pattern=on 2pt off 1.5pt, line width=.5pt] (3.0,0.17) -- (3.0,3.38);
  \node[fignum, cAdv, anchor=south, inner sep=1pt] at (3.0,3.34) {$a$};

  \fill[priorband] (0.5,2.47) rectangle (4.1,2.73);
  \fill[winband]   (2.6,1.32) rectangle (5.4,1.58);
  \fill[postband]  (2.6,0.17) rectangle (4.1,0.43);
  \draw[priorband, fill=none] (0.5,2.47) rectangle (4.1,2.73);
  \draw[winband,   fill=none] (2.6,1.32) rectangle (5.4,1.58);
  \draw[postband,  fill=none] (2.6,0.17) rectangle (4.1,0.43);

  \node[figtiny, anchor=west, inner sep=2.5pt] at (4.1,2.60) {$A_t$};
  \node[figtiny, anchor=west, inner sep=2.5pt, cWin]  at (5.4,1.45) {$W_t$};
  \node[figtiny, anchor=east, inner sep=2.5pt, cPost] at (2.6,0.30) {$A_{t+1}$};

  \node[seldot] at (3.2,2.60) {};
  \node[fignum, anchor=north, inner sep=2.5pt] at (3.28,2.47) {$\theta_t$};
  \node[advdot] at (3.0,0.30) {};

  \draw[meas] (0.5,2.93) -- (4.1,2.93) node[measlab, pos=0.22, above=7pt] {$L_t$};
  \draw[meas] (2.6,1.78) -- (5.4,1.78) node[measlab, pos=0.72, above=7pt] {$2/s_t$};
  \draw[meas] (2.6,0.63) -- (4.1,0.63) node[measlab, midway, above=7pt] {$L_{t+1}$};
\end{scope}

\begin{scope}[xshift=6.65cm]
  \draw[-{Latex[length=2.6mm,width=1.8mm]}, cPriorE, line width=1.1pt]
        (0.2,1.45) -- (1.45,1.45);
  \node[callout, anchor=south, inner sep=3pt] at (0.82,1.57)
        {$\varphi_t(\alpha)=(\alpha-\theta_t)x_t$};
  \node[callout, anchor=north, inner sep=3pt, text width=2.3cm] at (0.82,1.35)
        {recenter at $\theta_t$,\\ rescale by $x_t$};
\end{scope}

\begin{scope}[xshift=8.4cm]
  \node[figlab, anchor=south west] at (-0.15,3.60) {\pnl{b} regressor coordinate $\beta$};
  \draw[axisline] (-0.1,-0.35) -- (6.55,-0.35)
        node[fignum, above left=-1pt and -2pt] {$\beta$};
  \draw[tick] (1.4,-0.47) -- (1.4,-0.23);
  \node[fignum, anchor=north, inner sep=2pt] at (1.4,-0.47) {$0$};
  \draw[tick] (5.0,-0.47) -- (5.0,-0.23);
  \node[fignum, anchor=north, inner sep=2pt] at (5.0,-0.47) {$y=x_{t+1}$};

  \draw[gray!50, dashed, line width=.4pt] (4.0,0.43) -- (4.0,2.47);
  \draw[gray!50, dashed, line width=.4pt] (5.6,0.43) -- (5.6,2.47);
  \draw[cAdv, dash pattern=on 2pt off 1.5pt, line width=.5pt] (4.6,0.17) -- (4.6,3.38);
  \node[fignum, cAdv, anchor=south, inner sep=1pt] at (4.6,3.34) {$\varphi_t(a)$};

  \fill[priorband] (0.2,2.47) rectangle (5.6,2.73);
  \fill[winband]   (4.0,1.32) rectangle (6.0,1.58);
  \fill[postband]  (4.0,0.17) rectangle (5.6,0.43);
  \draw[priorband, fill=none] (0.2,2.47) rectangle (5.6,2.73);
  \draw[winband,   fill=none] (4.0,1.32) rectangle (6.0,1.58);
  \draw[postband,  fill=none] (4.0,0.17) rectangle (5.6,0.43);

  \node[figtiny, anchor=west, inner sep=2.5pt] at (5.6,2.60) {$J_t$};
  \node[figtiny, anchor=east, inner sep=2.5pt, cWin]  at (4.0,1.45) {window};
  \node[figtiny, anchor=west, inner sep=2.5pt, cPost] at (5.6,0.30) {$I$};

  \node[seldot] at (1.4,2.60) {};
  \node[fignum, anchor=north, inner sep=2.5pt] at (1.30,2.47) {$\theta_t$};
  \node[advdot] at (4.6,0.30) {};

  \draw[meas] (0.2,2.93) -- (1.4,2.93) node[measlab, midway, above=7pt]
       {near reach $q_-$};
  \draw[meas] (1.4,2.93) -- (5.6,2.93) node[measlab, pos=0.34, above=7pt] {far reach $q_+$};
  \draw[meas] (4.0,1.78) -- (6.0,1.78) node[measlab, midway, above=7pt] {$2$};
  \draw[meas] (4.0,0.63) -- (5.6,0.63) node[measlab, midway, above=7pt]
       {$\mu=s_tL_{t+1}$};
\end{scope}

\begin{scope}[yshift=-2.95cm, xshift=3.55cm]
  \node[figlab, anchor=south west] at (-0.05,1.28)
       {\pnl{c} the clipped case is an exact exchange: \ $y+\mu=q_+\,{+}\,1$};
  \fill[cAdv!22] (0,0) rectangle (5.2,0.34);
  \fill[cPostF]  (5.2,0) rectangle (7.1,0.34);
  \fill[white]   (7.1,0) rectangle (7.9,0.34);
  \draw[cPriorE, line width=.6pt] (0,0) rectangle (7.9,0.34);
  \draw[gray!70, line width=.4pt] (5.2,0) -- (5.2,0.34);
  \draw[gray!70, line width=.4pt] (7.1,0) -- (7.1,0.34);
  \draw[meas] (5.2,-0.18) -- (0,-0.18)
        node[measlab, midway, below=7pt] {realized magnitude $y$};
  \draw[meas] (7.1,-0.18) -- (5.2,-0.18)
        node[measlab, midway, below=7pt] {surviving uncertainty $\mu$};
  \node[measlab, anchor=west, black!55] at (7.95,0.17) {slack};
  \draw[meas] (7.9,0.52) -- (0,0.52)
        node[measlab, midway, above=7pt] {budget $q_+ + 1$};
\end{scope}

\end{tikzpicture}
\caption{The change of variables behind Corollary~\ref{cor:geom} and Proposition~\ref{prop:implicit}.
\pnl{a} In model space the update intersects the consistent set $A_t$ with a data
window $W_t$ whose width $2/s_t$ shrinks as the state grows.
\pnl{b} The affine map $\varphi_t$ sends the posited model $\theta_t$ to the
origin and rescales by the regressor $x_t$. The image $J_t$ of $A_t$ straddles
$0$ with reaches $q_\pm$, both equal to $g_t/2$ for the midpoint law, matching
the radius $r$ of Lemma~\ref{lem:window}, and at most $E_t s_t$ for a general
implicit selector, where $E_t := \sup_{\alpha \in A_t}\abs{\alpha - \theta_t}$ is
the selector's worst-case error; the window becomes an interval of
width exactly $2$ \emph{centered at the realized successor state} $y=x_{t+1}$;
and the posterior length is read off as $\mu = s_tL_{t+1}$. The window lemma, the
invariant, and the attack primitives of Appendix~\ref{app:echo} are all
statements about this one picture: how much of $J_t$ a width-$2$ window placed at
$y$ can leave behind.
\pnl{c} Once $y>1$ the window sits strictly on one side of $0$, and if it is
clipped by the far end of $J_t$ (the only way to realize a large $y$), the
exchange is exact: magnitude realized now is subtracted, unit for unit, from the
uncertainty left to exploit later.}
\label{fig:errorcoord}
\end{figure}
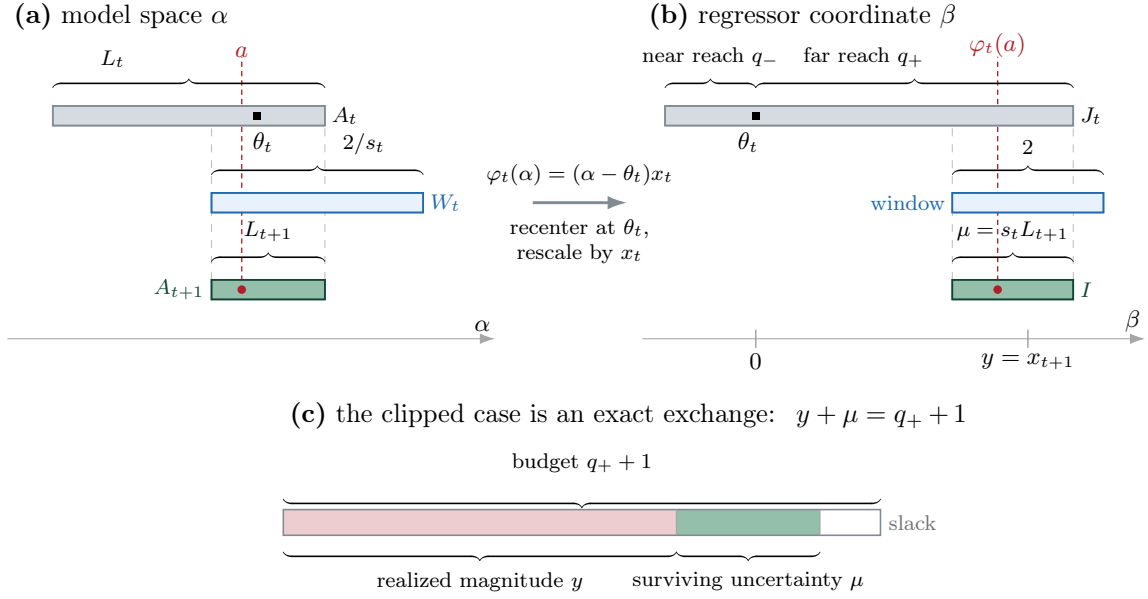

\begin{proposition}[Key invariant]\label{prop:invariant}
Under the controller of Definition~\ref{def:controller}, for every admissible $(a,w)$ and every $t \ge 0$:
\begin{equation}\label{eq:invariant}
    g_{t+1} \;\le\; L_t \;\le\; 2\Delta .
\end{equation}
\end{proposition}

\begin{proof}
$L_t \le 2\Delta$ is immediate from nestedness. For the first inequality, fix $t$. If $x_t = 0$: then $u_t = 0$, so $x_{t+1} = w_t$, $s_{t+1} \le 1$, and $L_{t+1} = L_t$; hence $g_{t+1} = L_t s_{t+1} \le L_t$. If $x_t \neq 0$: by Corollary~\ref{cor:geom}(ii), $\mu_{t+1}\abs{x_{t+1}} \le g_t$, whence
\[
    g_{t+1} \;=\; L_{t+1} s_{t+1} \;=\; \frac{\mu_{t+1}}{s_t}\,\abs{x_{t+1}} \;\le\; \frac{g_t}{s_t} \;=\; L_t . \qedhere
\]
\end{proof}

\begin{theorem}[Upper bound]\label{thm:ub}
The controller of Definition~\ref{def:controller} guarantees, for every $a \in [-\Delta, \Delta]$ and every $\norm{w}_\infty \le 1$, $\norm{x}_\infty \le 1 + \Delta$. Consequently $\gamma^\star(\Delta) \le 1 + \Delta$.
\end{theorem}

\begin{proof}
$\abs{x_0} = 0$ and $\abs{x_1} = \abs{w_0} \le 1$ (since $u_0 = 0$). For $t \ge 1$, \eqref{eq:reach} and Proposition~\ref{prop:invariant} give
$\abs{x_{t+1}} \le g_t/2 + 1 \le L_{t-1}/2 + 1 \le \Delta + 1$.
\end{proof}

Theorem~\ref{thm:ub} together with Proposition~\ref{prop:lb} proves Theorem~\ref{thm:main}, the infimum and suprema all being attained.

\begin{remark}[The bound restarts from any low state]\label{rem:restart}
The proof of Theorem~\ref{thm:ub} uses only $\abs{x_1} \le 1$ together with Proposition~\ref{prop:invariant}; the hypothesis $x_0 = 0$ enters solely to produce that bound. Hence the reusable form: if at any time $\tau$ the state satisfies $\abs{x_\tau} \le 1$ and the surviving set has half-width $L_\tau/2 \le \delta$, then $\abs{x_t} \le 1 + \delta$ for every $t \ge \tau$ --- immediately, since $\abs{x_{\tau+1}} \le g_\tau/2 + 1 \le \delta + 1$ and thereafter $\abs{x_{t+1}} \le L_{t-1}/2 + 1 \le \delta + 1$ by nestedness. This is the form invoked when the invariant is re-run from a post-probe consistent set in Propositions~\ref{prop:probing} and~\ref{prop:regret}, where $\delta$ is far smaller than $\Delta$.
\end{remark}

\begin{remark}[Anatomy of the worst case]\label{rem:anatomy}
Against the optimal controller the adversary's optimal play is: (i) use the disturbance to move the state from $0$ to $\pm1$ (this leaks nothing about $a$ because the state was $0$); (ii) immediately cash in with an extreme parameter $a = \pm\Delta$ aligned against the control and $w = \pm1$, producing the spike $\abs{x_2} = \Delta + 1$. The spike is self-defeating: by Lemma~\ref{lem:window}(i), realizing $\abs{x_2} = g_1/2 + 1$ forces $\mu_2 = 0$ --- the posterior collapses to a point, $a$ becomes known exactly, and thereafter deadbeat control keeps $\abs{x_t} \le 1$ forever. Any partial spike likewise buys the controller information at the exact rate \eqref{eq:tradeoff}; the invariant \eqref{eq:invariant} says the adversary can never accumulate more exposure than the prior diameter, no matter how it schedules the attack. A stealth adversary that keeps the state small teaches the controller nothing, but then no peak ever forms.
\end{remark}

\begin{remark}[Exactness, uniqueness, generalizations]\label{rem:extensions}
\emph{(a)} The value $1+\Delta$ is exact for every $\Delta > 0$, and the price of identification is exactly linear: $\gamma^\star(\Delta) - \gamma^\star(0) = \Delta$.
\emph{(b)} The optimal controller is not unique off worst-case paths, but every optimal controller must play $u_0 = 0$ and then $u_1 = 0$ (Remark~\ref{rem:u0}), and must respect the exposure cap and midpoint forcing of Proposition~\ref{prop:implicit} everywhere.
\emph{(c)} For an asymmetric prior $a \in [\underline a, \bar a]$ with diameter $D := \bar a - \underline a$, the value is $\gamma^\star = 1 + D/2$: the upper-bound argument applies verbatim (no step uses symmetry of the prior --- only that $A_t$ is an interval and $\hat a_t$ its midpoint), while in the lower bound the sign-aligned extreme is replaced by the endpoint whose response has larger magnitude: the two endpoint responses $\underline a\, x_1 + u_1$ and $\bar a\, x_1 + u_1$ differ by $D\abs{x_1} \ge D$, so one of them has magnitude at least $D/2$. The value thus depends only on the \emph{diameter} of the parameter uncertainty, not on how unstable the nominal system is.
\end{remark}

\subsection{The exposure cap is necessary for every controller}\label{sec:cap}

The invariant \eqref{eq:invariant} was derived for the law \eqref{eq:law}. We close the section by showing that its conclusion is a constraint on \emph{every} optimal controller, and that it follows from one aligned attack, with no window lemma at all. For this one proposition we suspend the standing assumption of the section and let $K \in \Kclass$ be arbitrary; $A_t$, $s_t$, $L_t$ and $g_t = L_t s_t$ are read off the realized trajectory as before, $\theta_t = -u_t/x_t$ is the implicit selector of Definition~\ref{def:implicit}, and
\[
    E_t \;:=\; \sup_{\alpha \in A_t}\abs{\alpha - \theta_t}
\]
is that selector's worst-case error (for the law \eqref{eq:law} it equals $L_t/2$, the bound on $\abs{e_t}$ in \eqref{eq:closedloop}). Write $\gamma(K) := \sup_{a,\,\norm{w}_\infty \le 1}\norm{x(K,a,w)}_\infty$. A history is \emph{reachable under $K$} if it is realized by $K$ together with some admissible pair $(a, w)$.

\begin{proposition}[The exposure cap is necessary]\label{prop:implicit}
Let $K \in \Kclass$. At every history reachable under $K$ with $x_t \neq 0$:
\begin{enumerate}[label=(\roman*), nosep]
    \item $\displaystyle \gamma(K) \;\ge\; 1 + E_t\,s_t \;\ge\; 1 + \frac{g_t}{2}$, the second inequality holding with equality if and only if $\theta_t = \midpt(A_t)$;
    \item hence if $K$ is optimal, $\gamma(K) = 1 + \Delta$, then $E_t s_t \le \Delta$ and in particular $g_t \le 2\Delta$: the cap of Proposition~\ref{prop:invariant}, now as a necessary condition on \emph{every} optimal controller;
    \item and at a history of maximal exposure $g_t = 2\Delta$ the implicit selector is forced, $\theta_t = \midpt(A_t)$, i.e.\ $u_t = -\midpt(A_t)\,x_t$: a deviation $d := \abs{\theta_t - \midpt(A_t)} > 0$ costs at least $d\,s_t$, in that $\gamma(K) \ge 1 + \Delta + d\,s_t$.
\end{enumerate}
\end{proposition}

\begin{proof}
\emph{(i)} Fix the observed prefix. Every $\alpha \in A_t$ is consistent with it: taking $w_s^\alpha := x_{s+1} - \alpha x_s - u_s$ for $s < t$ gives $\abs{w_s^\alpha} \le 1$ because $\alpha \in W_s$, and these disturbances reproduce the prefix under parameter $\alpha$. Since $K$ is causal and deterministic and the prefix is the same, $K$ plays the same $u_t$, hence the same $\theta_t$, under every such $\alpha$. The adversary therefore \emph{defers its commitment}: it takes $a^\star \in \arg\max_{\alpha \in A_t}\abs{\alpha - \theta_t}$, an endpoint of the closed interval $A_t$ (so the maximum is attained), fixes the corresponding prefix disturbances $w_s^{a^\star}$, and aligns $w_t := \sgn\big((a^\star - \theta_t)x_t\big)$, giving
\[
    \abs{x_{t+1}} \;=\; \abs{(a^\star - \theta_t)x_t + w_t} \;=\; E_t\,s_t + 1 .
\]
For the second inequality,
\[
    E_t \;=\; \max\big(\abs{\theta_t - \min A_t},\, \abs{\theta_t - \max A_t}\big) \;\ge\; \tfrac12\big(\abs{\theta_t - \min A_t} + \abs{\theta_t - \max A_t}\big) \;\ge\; L_t/2,
\]
with equality throughout exactly when $\theta_t$ is the midpoint $\midpt(A_t)$.
\emph{(ii)} is \emph{(i)} evaluated at $\gamma(K) = 1+\Delta$.
\emph{(iii)} For an interval, $E_t = L_t/2 + \abs{\theta_t - \midpt(A_t)}$ (the worst corner is the one the selector moved away from), so $g_t = 2\Delta$ gives $E_t s_t = \Delta + d\,s_t$, and \emph{(i)} applies.
\end{proof}

Two readings. First, Proposition~\ref{prop:invariant} and Proposition~\ref{prop:implicit}(ii) are converses: the window lemma says the exposure cap $g_t \le 2\Delta$ can be \emph{maintained}; the deferred-commitment attack says it must be. The cap is a law of the problem, not a property of the law \eqref{eq:law}. (Only the cap is forced: the recursive form $g_{t+1} \le L_t$ of \eqref{eq:invariant} is a property of the midpoint law, and an optimal controller with submaximal exposure may violate it within the cap.) Second, the forcing in \emph{(iii)} is genuinely local, and quantifiably so: away from maximal exposure it leaves the tolerance
\[
    \abs{\theta_t - \midpt(A_t)} \;\le\; \frac{2\Delta - g_t}{2\,s_t},
\]
an outer bound which is strictly positive as soon as the exposure is submaximal, as it is at every step following a maximal-exposure one.\footnote{At such a step $s_t \ge 1$ (since $g_t = 2\Delta$ and $L_t \le 2\Delta$) and $\theta_t = \midpt(A_t)$ is forced, so Corollary~\ref{cor:geom} applies: $g_{t+1} \le \mu_{t+1}\abs{x_{t+1}}$ with $\mu_{t+1} \le \min(2, g_t)$ and $\abs{x_{t+1}} + \mu_{t+1} \le g_t/2 + 1$. Maximizing $\mu_{t+1}\abs{x_{t+1}}$ under these constraints gives $2(\Delta - 1)$ for $\Delta \ge 3$, $(\Delta+1)^2/4$ for $\tfrac13 \le \Delta \le 3$, and $2\Delta(1-\Delta)$ for $\Delta \le \tfrac13$; in every case strictly below $2\Delta$: a maximal-exposure history is never followed by another one.} The optimal controller is thus pinned exactly on the maximal-exposure set $\{g_t = 2\Delta\}$, and constrained elsewhere by a tolerance that closes linearly as the exposure approaches its cap. The tolerance is necessary but not sufficient: a deviation also reshapes the posterior, and the adversary can re-aim the surviving uncertainty at the same bias --- the echo of Section~\ref{sec:exploration}.

\begin{remark}[The post-kick deviation penalty is exactly $d$]\label{rem:u1exact}
Remark~\ref{rem:u0} claims that at the post-kick history a deviation of size $d$ costs \emph{exactly} $d$. The lower bound is Proposition~\ref{prop:implicit}(iii); the matching upper bound takes two lines. Deviate by $d$ at the post-kick history --- at every other history play the midpoint law, so an adversary that does not kick fully ($s_1 < 1$) faces the unmodified law and Theorem~\ref{thm:ub} caps its peak at $1 + \Delta$. Against the full kick, the one-step reach is $1 + \Delta + d$; for the echo step, the image of $A_1$ in the regressor coordinate of Corollary~\ref{cor:geom} is an interval with reaches $\Delta \pm d$, contained in $[-(\Delta+d),\, \Delta+d]$, and enlarging the interval only enlarges the intersection, so Lemma~\ref{lem:window}(ii) gives $g_2 \le 2(\Delta + d)$ and $\abs{x_3} \le 1 + \Delta + d$. From $t = 2$ on the law is the midpoint law and the invariant applies step by step (as in Remark~\ref{rem:restart}) with $L_2 \le 2$ (the data window has width $2$ at $s_1 = 1$), so every later state is at most $1 + \min(1, \Delta) \le 1 + \Delta + d$. Verified over a deferred-commitment adversary search in Section~\ref{sec:numerics}.
\end{remark}

\begin{corollary}[Nonzero initial conditions]\label{cor:nonzeroic}
Fix $s_0 := \abs{x_0} \ge 1$ and run the update \eqref{eq:setmem} from the start, so that $A_1 = A_0 \cap W_0$. Then for every $\Delta > 0$ the value of the game from $x_0$ is $\max\big(s_0,\; 1 + \Delta s_0\big)$, and the law \eqref{eq:law} attains it. Unless the game degenerates --- the given start already being the minimax peak, $s_0 > 1 + \Delta s_0$, which requires $\Delta < 1 - 1/s_0$ and thus never occurs for $\Delta \ge 1$ --- every optimal controller plays $\theta_0 = \midpt(A_0)$, i.e.\ $u_0 = 0$, and a deviation of size $d$ costs at least $d\,s_0$.
\end{corollary}

\begin{proof}
\emph{Lower bound.} The start contributes $s_0$ to the peak outright, and $u_0 = K_0(x_0)$ is a fixed number, so Step~2 of Proposition~\ref{prop:lb} applies verbatim with $x_0$ in place of $x_1$ --- the initial condition already supplies the regressor that the kick there buys: the aligned extreme parameter and disturbance realize $\abs{x_1} = \Delta s_0 + \abs{u_0} + 1 \ge 1 + \Delta s_0$.
\emph{Upper bound.} By \eqref{eq:reach}, $\abs{x_1} \le g_0/2 + 1 = \Delta s_0 + 1$; Proposition~\ref{prop:invariant}, whose proof nowhere uses $x_0 = 0$, applies from $t = 0$ and caps every later state by $\Delta + 1 \le \Delta s_0 + 1$.
\emph{Forcing.} If $1 + \Delta s_0 \ge s_0$, optimality and Proposition~\ref{prop:implicit}(i) at $t = 0$ give $1 + E_0 s_0 \le 1 + \Delta s_0$, so $E_0 \le \Delta = L_0/2$, attained only at $\theta_0 = \midpt(A_0)$; a deviation of size $d$ makes $E_0 = \Delta + d$.
\end{proof}

The forced action is again $u_0 = 0$, now by symmetry of the prior rather than by passivity. In the degenerate case the forcing relaxes, by Proposition~\ref{prop:implicit}(i), to the tolerance $\abs{\theta_0} \le (s_0 - 1)/s_0 - \Delta$, and the tolerance is genuine: any static gain within it is also optimal from such a start, as it must be, since for $\Delta < 1$ adaptation is optional (Remark~\ref{rem:adapt}). Nothing deteriorates with the initial condition: the value stays linear in $s_0$, one spike deep, and the architecture is unchanged.

\section{Why probing, commitment, and optimism are suboptimal}\label{sec:exploration}

The optimal policy of Definition~\ref{def:controller} contains none of the mechanisms that the stochastic adaptive control and reinforcement learning traditions treat as essential: no probing or dither signal, no persistency-of-excitation condition, no optimism, no explicit identification phase, no regret-driven balancing of exploration against exploitation. This section proves that in the adversarial setting this is not an idiosyncrasy of one clever policy --- each of those mechanisms is \emph{structurally} suboptimal for the problem \eqref{eq:minimax}. The common cause is an information-timing asymmetry:
\begin{quote}
\emph{Exposure is paid at time $t$; the information it buys arrives at time $t+1$; and the adversary chooses when exposure is expensive.}
\end{quote}
In stochastic settings a probe's cost is a constant amortized against future gain; here, the adversary aligns the unknown parameter and disturbance against the probe at the moment of maximal uncertainty, so the probe's cost is multiplied by $\Delta$ \emph{before} its information can be used.

The scope of these results is the objective \eqref{eq:minimax}: peak cost, full actuation, no input penalty, initial state at rest. We do not claim that exploration is universally counterproductive in adversarial games. Under quadratic accumulated-cost criteria, minimax-optimal policies can themselves be dual, trading exploration against exploitation, even with full actuation and no input penalty \citep{rantzer2025actuated}, and likewise for a single unknown input direction under an input-weighted criterion \citep{rantzer2026single}. The claim here is both narrower and sharper: for the peak objective every standard mechanism pays a penalty that the amortization intuition underestimates, and that penalty can be bounded below explicitly. Matching upper and lower bounds are supplied for the \emph{optimal} law at every $\Delta$ (Theorem~\ref{thm:main}); for the suboptimal alternatives priced in this section we give lower bounds only (the probing penalty, which is attained exactly, being the one exception), since the point is the separation, not the constant.

\subsection{Probing is punished before it pays}

\begin{proposition}[Linear-in-$\Delta$ probing penalty]\label{prop:probing}
Let $K \in \Kclass$ with $u_0 = K_0(0)$ and $\eta := \abs{u_0}$. Then
\[
    \sup_{a,\,\norm{w}_\infty \le 1} \norm{x(K,a,w)}_\infty \;\ge\; \max\Big(1 + \eta,\;\; 1 + (1+\eta)\,\Delta\Big),
\]
and the bound is tight: superimposing the probe on the optimal law of Definition~\ref{def:controller} attains it exactly. In particular, any controller that explores with a nonzero input before information exists forfeits optimality by at least $\eta\Delta$: the probe is amplified by the full parameter uncertainty one step before its information becomes available.
\end{proposition}

\begin{proof}
Run the construction of Proposition~\ref{prop:lb}: the aligned kick gives $\abs{x_1} = \eta + 1$, so the peak is at least $\eta+1$; no information is revealed (the regressor is $x_0 = 0$), and the aligned extreme parameter and disturbance give $\abs{x_2} = \Delta\abs{x_1} + \abs{u_1} + 1 \ge 1 + (1+\eta)\Delta$. For tightness, the probe-then-midpoint law has post-probe exposure at most $2(1+\eta)\Delta$, and the invariant \eqref{eq:invariant} (applied from $t \ge 1$, where the law coincides with Definition~\ref{def:controller}) caps every later peak at $1 + \Delta$.
\end{proof}

The statement time-shifts: if the adversary can steer play to a history with $x_t = 0$ while the consistent set still has length $L$ (e.g., by playing $w \equiv 0$ until the controller first departs from zero at a zero state --- along which no information ever leaks), then a probe of size $\eta$ at that point is punished by at least $\max\big(1+\eta,\, 1 + (1+\eta)L/2\big)$. Consequently every optimal controller is passive at \emph{every} all-zero history, not merely at $t=0$.

\subsection{Explore-then-commit is uniformly dominated}\label{sec:etc}

Proposition~\ref{prop:probing} penalizes individual probes. The canonical two-phase design (\emph{identify first, then apply certainty-equivalent control with the frozen estimate}), still standard practice \citep{hazan2019nonstochastic}, fails more broadly, because the adversary can sabotage either phase at will: starve the identification phase of information, or monetize its excitation. The instrument for the first option is a \emph{cancellation adversary}: as long as the state is zero, any input of magnitude at most $1$ can be met with $w_t = -u_t$, so the state remains at zero, every regressor is zero, \emph{no information about $a$ ever leaks} --- and, the trajectory being independent of $a$, the adversary retains complete freedom in its later choice of parameter.

\begin{proposition}[Explore-then-commit impossibility]\label{prop:etc}
Let $\Delta \ge 1$, and let $K \in \Kclass$ (deterministic, as throughout) have the \emph{commitment property}: there exist a stopping time $T$ of the closed-loop information filtration, finite on every admissible trajectory, and an $\mathcal{F}_T$-measurable gain $\hat a$, such that $u_t = -\hat a\,x_t$ for all $t \ge T$ \emph{on every admissible continuation of the stopped prefix}. Then:
\begin{enumerate}[label=(\roman*), nosep]
    \item the worst-case peak satisfies $\displaystyle\sup_{a,\,\norm{w}_\infty\le1} \norm{x}_\infty \;\ge\; 1 + 2\Delta \;=\; \gamma^\star(\Delta) + \Delta$, i.e., commitment at least \emph{doubles} the optimal excess over the disturbance floor;
    \item if, in addition, $K$ never applies an input of magnitude exceeding $1$ at a zero state, then the worst-case peak is \emph{infinite}.
\end{enumerate}
In particular, the standard fixed-horizon explore-then-commit scheme (any identification phase of deterministic length, followed by certainty-equivalent control with any frozen estimate) is strictly suboptimal for every $\Delta \ge 1$, and infinitely so unless its exploration injects inputs larger than the disturbance bound.
\end{proposition}

\begin{proof}
Run the cancellation adversary from $t = 0$: while the state is zero and $\abs{u_t} \le 1$, play $w_t = -u_t$. Along this play all states are zero, all regressors are zero, and the consistent set remains the full prior; since the trajectory does not involve $a$, it (and everything the deterministic $K$ does on it) is fixed in advance.

\emph{Branch 1: $K$ never exceeds magnitude $1$ at a zero state.} Then the cancellation continues forever, the trajectory is identically zero, and the stopping time is attained at a fixed finite $T$ on this trajectory, with a fixed $\hat a$; because commitment binds on every continuation of the stopped prefix, all play from $T$ on, including after the kick below, is $u_t = -\hat a\,x_t$. Since no information leaked, $\hat a$ is a number known in advance, and some admissible $a^\dagger$ satisfies $\abs{a^\dagger - \hat a} \ge \Delta \ge 1$. The adversary commits to $a^\dagger$, plays $w_T = \pm 1$, and thereafter aligns $w_t$ with $(a^\dagger - \hat a)x_t$: the closed loop obeys $\abs{x_{t+1}} = \abs{a^\dagger - \hat a}\,\abs{x_t} + 1 \ge \abs{x_t} + 1$, which is unbounded. This proves (ii), and (i) on this branch.

\emph{Branch 2: at some first time $\tau$, $K$ plays $\abs{u_\tau} = \eta > 1$ at a zero state.} Up to $\tau$ the cancellation ran, so the consistent set is still the full prior and $u_\tau$, being determined by the (fixed, $a$-independent) history, is known in advance. The adversary aligns $w_\tau = \sgn(u_\tau)$, giving $\abs{x_{\tau+1}} = \eta + 1 > 2$ with a zero regressor, still no information, and then plays the kick-and-cash step of Proposition~\ref{prop:lb}: an aligned extreme $a^\star = \pm\Delta$ and aligned $w_{\tau+1}$ give
\[
    \abs{x_{\tau+2}} \;\ge\; \Delta\,(\eta + 1) + 1 \;>\; 1 + 2\Delta. \qedhere
\]
\end{proof}

Three readings of this result. First, it closes the loophole left by Proposition~\ref{prop:probing}: a two-phase scheme cannot evade the probing penalty by keeping its excitation small, because sub-disturbance-level excitation is \emph{erasable} --- the adversary cancels it exactly and the identification phase ends with nothing. Second, the escape routes from the dichotomy are the abandonments of the two-phase design: a scheme that commits only ``once the uncertainty is small'' never commits on the starved trajectory (it has become an ongoing adaptive law); a scheme that keeps updating after commitment is not committed. What remains after excluding both failure branches is indefinite, passive, consistency-driven adaptation: the architecture of Definition~\ref{def:controller}. Third, the bound in (i) is a floor, not the typical price: a schedule that keeps its excitation at or below the disturbance level to limit its cost falls into branch (ii), and its worst-case peak is infinite.

\begin{remark}[The binding requirement is sharp]\label{rem:binding}
Commitment must be an event of the observed prefix, binding on all admissible continuations; a merely per-trajectory version of the property (``on every trajectory the tail happens to be linear in the state'') is strictly weaker, and the proposition fails for it. Concretely, the controller that plays zero at rest, responds to the \emph{first} nonzero state with a probe of size $\Delta + 3$, and commits two steps later to the center of the width-$\le 1$ consistent set produced by the post-probe regressor achieves a \emph{finite} worst-case peak while satisfying the per-trajectory property. Its fate still illustrates the theorem's economics: the probe is amplified by the still-untouched uncertainty before its information arrives, and the worst-case peak is $\Theta(\Delta^2)$ --- at most $\Delta(2\Delta+4)+1$ by the window argument, and at least $\Delta(\Delta+3)+1$ by an explicit attack (present a first state $\varepsilon \le \tfrac{1}{2\Delta}$, whose zero regressor leaks nothing and whose probe transition retains the full prior, so the interim certainty-equivalent step plays $u = 0$; then align $a = \Delta$, $w = 1$ against the probe state).
\end{remark}

This is the scalar form of what is plausibly a general phenomenon: a probe large enough to dominate the adversary's realized deviations is itself a state of the same magnitude as the exposure it is trying to prevent. Deliberate excitation does not escape the identification spike --- it feeds it. Note also what the proposition does \emph{not} say: it does not claim learning is useless (the problem is unsolvable without it, Remark~\ref{rem:adapt}); it says that in the adversarial setting all the information worth paying for is delivered \emph{free of charge} by the adversary's own attack, at the moment the adversary attacks.

\subsection{The price of biased selection}

By Definition~\ref{def:implicit} the selector is the only design freedom there is: every causal law is $u_t = -\theta_t x_t$ for its implicit $\theta_t$. So we price selection directly (what does it cost to select anything other than the midpoint?), first for a single step, then for the standing bias that optimism produces.

\begin{proposition}[One-step selection penalty]\label{prop:endpoint}
Let $K \in \Kclass$ be any causal controller and let $\theta_1 = -u_1/x_1$ be its implicit selector at the post-kick history, where $\abs{x_1} = 1$ and $A_1 = [-\Delta, \Delta]$. Then
\[
    \gamma(K) \;\ge\; 1 + \sup_{\alpha \in A_1}\abs{\alpha - \theta_1} \;\ge\; 1 + \Delta,
\]
the second inequality strict unless $\theta_1 = 0$, i.e.\ unless $K$ plays $u_1 = 0$. In particular, if the implicit selector sits at relative position $\lambda \in [0,1]$ of the consistent interval ($\theta_1 = \min A_1 + \lambda\abs{A_1}$), then
\[
    \gamma(K) \;\ge\; 1 + 2\Delta\,\max(\lambda,\, 1-\lambda),
\]
and $\lambda = \tfrac12$ is the unique minimizer, achieving the optimal $1 + \Delta$.
\end{proposition}

\begin{proof}
The history in the statement is reachable: the kick $w_0 = \pm 1 - u_0$ is admissible for at least one sign whenever $\abs{u_0} \le 2$, realizes $\abs{x_1} = 1$, and, the regressor being $x_0 = 0$, leaks nothing, so $A_1 = A_0$. (For $\abs{u_0} > 2$ no such history exists and the claim is vacuous; those laws pay $\ge 1 + 3\Delta$ by Proposition~\ref{prop:probing} in any case.) Apply Proposition~\ref{prop:implicit}(i) at $t = 1$, where $s_1 = 1$ and $L_1 = 2\Delta$. For the second form, $\sup_{\alpha \in A_1}\abs{\alpha - \theta_1} = 2\Delta\max(\lambda, 1-\lambda)$.
\end{proof}

One step already separates every non-midpoint controller from the optimum. A \emph{standing} bias costs more than that one-step toll, because the adversary can collect against it twice: rather than cash in at the spike, it may hold back, leaving a posterior positioned to face the same bias again. For endpoint (optimistic-style) selection this doubles the damage.

\begin{proposition}[The echo: endpoint selection pays at least $4\Delta - 1$]\label{prop:echo}
Let $K^{\mathrm{end}}$ be the certainty-equivalent deadbeat law whose selector always picks an endpoint of $A_t$ (either endpoint, by any measurable rule). Then
\[
    \gamma\big(K^{\mathrm{end}}\big) \;\ge\; 4\Delta - 1 \qquad\text{for every }\Delta > 0,
\]
asymptotically at least \emph{four} times the optimal $1 + \Delta$. For $\Delta \ge 3/2$ the bound is realized by a two-stage attack: a \emph{half-spike}, a spike tuned below maximal height so that the width-$2$ data window parks wholly inside the prior at the far end from the selected endpoint (peak $2\Delta - 1$, posterior length $2$), followed by an \emph{echo} in which whichever endpoint the selector now picks is off by the full posterior length $2$, doubling the state once more. More generally $\gamma(K^{\mathrm{end}}) \ge 1 + h(2\Delta)$ for $\Delta \ge 1/2$ and $\gamma(K^{\mathrm{end}}) \ge 1 + 2\Delta$ below, where
\[
    h(\ecap) := \big(\tfrac{\ecap+1}{2}\big)^2 \ \ (\ecap \le 3), \qquad h(\ecap) := 2(\ecap-1) \ \ (\ecap \ge 3),
\]
with $\ecap := 2\Delta$ the \emph{error cap}, a bound on the selector's worst-case error, here $E_t = L_t \le 2\Delta$ and kept general in Appendix~\ref{app:echo}; the optimal tuning of the spike height at every cap is Lemma~\ref{lem:lb}, and the display follows since $1 + h(2\Delta) - (4\Delta - 1) = (2\Delta - 3)^2/4 \ge 0$ for $\tfrac12 \le \Delta \le \tfrac32$ and $1 + 2\Delta \ge 4\Delta - 1$ for $\Delta \le \tfrac12$.
\end{proposition}

\begin{proof}
Kick to $\abs{x_1} = 1$; realize $x_2 = \mp(2\Delta - 1)$ with signs chosen so that the induced window $[\,a' - 1, a' + 1\,]$, centered at $a' := (x_2 - u_1)/x_1$, lies inside $A_1$ at the far end from the selected endpoint $e_1$ (feasible for $\Delta \ge 3/2$; the reconstructed disturbance satisfies $\abs{w_1} \le 1$); the posterior $A_2$ has length $2$ and contains the prior's far endpoint, so whichever endpoint $e_2$ the selector picks, the opposite end $a$ of $A_2$ satisfies $\abs{a - e_2} = 2$, and the aligned $w_2$ gives $\abs{x_3} = 2\abs{x_2} + 1 = 4\Delta - 1$. Appendix~\ref{app:echo} names the three moves (Definition~\ref{def:attacks}), replays this run step by step at generic spike height (Table~\ref{tab:run}), and optimizes the height at every error cap (Lemma~\ref{lem:lb}), yielding the piecewise bound.
\end{proof}

The echo phenomenon shows that a selection bias is not a one-time toll but a standing vulnerability: the adversary can \emph{re-position} the surviving uncertainty to face the bias again.

\subsection{Optimism cannot select: a dichotomy}\label{sec:ofudichotomy}

Optimism in the face of uncertainty (select the model in the confidence set under which the achievable cost is smallest, then act on it) is the canonical principled answer to the selection problem, and the dominant exploration mechanism of the learning-for-control literature \citep{abbasi2011regret, cohen2019learning}. Instantiated on this problem it faces a dichotomy: either the ranking criterion is the problem's own objective, in which case it induces no ranking at all and selection is pure tie-breaking (Proposition~\ref{prop:ofu}); or it genuinely ranks, in which case it anchors at a preferred model and pays asymptotically at least twice the optimum (Proposition~\ref{prop:ofudichotomy}). The first branch:

\begin{proposition}[Value function is model-independent]\label{prop:ofu}
For the known-parameter game (fixed $a$, adversarial $\norm{w}_\infty \le 1$), the optimal worst-case peak from state $x$ onward (the supremum over the current and all future states) is
\[
    V(x; a) \;=\; \max\{\abs{x},\, 1\} \qquad\text{for every } a \in \R,
\]
attained by the deadbeat law $u = -ax$. Consequently, every model in the consistent set is exactly as ``optimistic'' as every other: the OFU selection rule induces no ranking whatsoever on $A_t$, and its output is determined entirely by its tie-breaking rule.
\end{proposition}

\begin{proof}
Under $u = -ax$, $x' = w$, so the future peak is at most $\max\{\abs{x}, 1\}$. Conversely the current state contributes $\abs{x}$, and for any $u$ the adversary plays $w = \sgn(ax+u)$, giving $\abs{x'} = \abs{ax+u} + 1 \ge 1$; hence $V(x;a) \ge \max\{\abs{x},1\}$.
\end{proof}

The degeneracy has a structural cause: full actuation with an unconstrained, uncosted input. With the model known, one deadbeat step cancels the dynamics entirely, so all models are one-step-equivalent going forward; the stochastic-LQR counterpart is the \emph{uninformative optimal policy} of \citet{ziemann2021uninformative}, under which optimal play generates no information about the unknown parameters and logarithmic regret becomes impossible. Whether optimism regains content under partial actuation is an open question outside our scope. The peak objective is what makes the degeneracy fatal rather than benign: under an accumulated-cost criterion a minimax-optimal policy may still probe (in the fully actuated game with no input penalty of \citet{rantzer2025actuated} it provably does), whereas here probing is punished before it pays (Propositions~\ref{prop:probing} and~\ref{prop:etc}). With the problem's own objective, then, optimism is pure tie-breaking, and its adversarially worst tie-break pays at least $4\Delta - 1$, asymptotically four times the optimum: the corner tie-break reproduces the endpoint law of Proposition~\ref{prop:echo}. One might hope that a \emph{richer} optimism (scoring models by a cost that also penalizes input effort, or discounts the future) breaks the tie and rescues the principle. It breaks the tie; it does not rescue the principle:

\begin{proposition}[Anchored-optimism penalty]\label{prop:ofudichotomy}
Let $J: [-\Delta, \Delta] \to \R$ be strictly quasi-convex with minimizer $a^\circ$ (e.g.\ $J(a) = $ the optimal achievable cost under model $a$, for any objective increasing in control effort, in which case $a^\circ = 0$), and let $K_J$ be the certainty-equivalent deadbeat law with the optimistic selector $\theta_t \in \arg\min_{\theta \in A_t} J(\theta)$. Then $\theta_t$ is the point of $A_t$ nearest $a^\circ$ (the \emph{anchor}), and, with $R := \Delta + \abs{a^\circ}$ the distance from the anchor to the far corner of the prior,
\[
    \gamma(K_J) \;\ge\; 2R - 1 \;=\; 2\Delta + 2\abs{a^\circ} - 1 \qquad (R \ge 3).
\]
For the canonical input-penalized criterion $J(a) = \abs{a}$ this gives $\gamma(K_{\abs{\cdot}}) \ge 2\Delta - 1$ for $\Delta \ge 3$, asymptotically at least \emph{twice} the optimum. The bound unifies the selection penalties of this section: anchoring at a prior corner ($\abs{a^\circ} = \Delta$) recovers the endpoint law's $4\Delta - 1$; anchoring at the center ($a^\circ = 0$) still pays $2\Delta - 1$; no anchor is safe.
\end{proposition}

\begin{proof}
Strict quasi-convexity makes $\arg\min_{A} J$ the clamp of $a^\circ$ onto the interval $A$. Assume $a^\circ \ge 0$ (else mirror) and attack the far corner $-\Delta$; the run is tabulated in Table~\ref{tab:anchor}. The kick $w_0 = -1$ gives $x_1 = -1$ and leaks nothing; the selector plays $\theta_1 = a^\circ$, $u_1 = -\theta_1 x_1 = a^\circ$. The adversary commits to $a = -\Delta$ and plays the \emph{half-spike} $w_1 = -1$:
\[
    x_2 \;=\; a\,x_1 + u_1 + w_1 \;=\; \Delta + a^\circ - 1 \;=\; R - 1,
\]
and the induced window $\{\alpha : \abs{x_2 - \alpha x_1 - u_1} \le 1\} = [\,a^\circ - R,\; a^\circ - R + 2\,] = [-\Delta,\, -\Delta + 2]$: the posterior has length $2$ and sits at the far corner. The selector clamps to its near end, $\theta_2 = -\Delta + 2$, at distance $2$ from the true $a$; aligning $w_2 = -1$ gives $\abs{x_3} = 2\abs{x_2} + 1 = 2R - 1$.
\end{proof}

\begin{table}[t]
\centering
\footnotesize
\setlength{\tabcolsep}{3.5pt}
\begin{tabular}{@{}cl c ccc cc cc c@{}}
\toprule
&& state & \multicolumn{3}{c}{consistent set} & \multicolumn{2}{c}{controller} & \multicolumn{2}{c}{adversary} & transition \\
\cmidrule(lr){3-3}\cmidrule(lr){4-6}\cmidrule(lr){7-8}\cmidrule(lr){9-10}\cmidrule(l){11-11}
$t$ & move & $x_t$ & $A_t$ & $L_t$ & $E_t$ & $\theta_t$ & $u_t$ & $a$ & $w_t$ & window $W_t$ \\
\midrule
$0$ & kick       & $0$    & $[-\Delta,\Delta]$      & $2\Delta$ & --- & --- & $0$ & free & $-1$ & $\R$ \\
$1$ & half-spike & $-1$   & $[-\Delta,\Delta]$      & $2\Delta$ & $R$ & $a^\circ$ (clamp) & $a^\circ$ & $:= -\Delta$ & $-1$ & $[-\Delta,\; -\Delta{+}2]$ \\
$2$ & echo       & $R{-}1$ & $[-\Delta,\, -\Delta{+}2]$ & $2$   & $2$ & $-\Delta{+}2$ (clamp) & $(\Delta{-}2)(R{-}1)$ & $-\Delta$ & $-1$ & --- \\
$3$ & ---        & $-(2R{-}1)$ & \multicolumn{8}{l}{peak $= 2R - 1 = 2\Delta + 2\abs{a^\circ} - 1$} \\
\bottomrule
\end{tabular}
\caption{The anchored-optimism attack of Proposition~\ref{prop:ofudichotomy}, mirrored so that $a^\circ \ge 0$; $R = \Delta + \abs{a^\circ}$. Row $t$ is the game state at step $t$: $A_t$ is the consistent set entering the step, $L_t$ its length, $E_t = \sup_{\alpha \in A_t}\abs{\alpha - \theta_t}$ the selector's worst-case error, and $W_t$ the data window induced by the transition out of the step ($A_{t+1} = A_t \cap W_t$). Unlike the endpoint attack of Proposition~\ref{prop:echo}, the selector is an \emph{interior} clamp ($E_1 = R < L_1$) and no deferral is needed: the adversary commits to $a = -\Delta$ already at the spike, since the clamp's later selections are determined. The half-spike parks the width-$2$ window at the far corner; the clamp then sits at distance $2$ from $a$, and the echo doubles the state.}
\label{tab:anchor}
\end{table}

\begin{remark}[The midpoint is not the optimizer of any model-wise criterion]\label{rem:setwise}
No per-model criterion that expresses a strict preference between some two interior models can select the midpoint on every consistent interval: if $\midpt(A) \in \arg\min_{A} J$ for every interval $A \subseteq [-\Delta, \Delta]$, then for every interior point $y$ and every $\varepsilon \in (0,\, \Delta - \abs{y}]$ the interval $A = [y-\varepsilon, y+\varepsilon]$ forces $J(y) \le J(z)$ for all $z \in A$, and chaining overlapping such intervals (with steps bounded below on compact subsets of the interior) makes $J$ constant on the interior of the prior. (Criteria constant on the interior --- with corner values $J(\pm\Delta)$ no smaller than the interior constant, since otherwise the criterion selects the \emph{corners} and reproduces the endpoint law of Proposition~\ref{prop:echo} --- are the only escape, and for them the selection is again pure tie-breaking, returning to the first branch of the dichotomy.) The midpoint is a \emph{set-relative} object: $\arg\min_\theta \max_{a \in A} \abs{a - \theta}$, the Chebyshev center (the one-dimensional Steiner point) of the surviving uncertainty. Optimism ranks worlds; the security (minimax) strategy hedges over them; and in the adversarial setting only the hedge is optimal.
\end{remark}

Randomized selection --- posterior (Thompson) sampling over the consistent set \citep{abeille2017thompson, ouyang2017control} --- escapes neither branch of the dichotomy. That no randomized controller beats $\gamma^\star$ is already Remark~\ref{rem:randomization}; Proposition~\ref{prop:endpoint}, applied to the random implicit selector at the critical step, says what sampling must do to survive it. Because the parameter is chosen without sight of the realization, the bound there reads $\max_{\sigma = \pm 1}\mathbb{E}\abs{\sigma\Delta - \theta_1} \ge \Delta$, with equality if and only if $\theta_1$ is supported in $A_1$ with $\mathbb{E}[\theta_1] = \midpt(A_1)$: a sampler survives the critical step precisely by reproducing the midpoint \emph{in mean}, and any distribution tilted toward an extreme, as optimism-weighted sampling is, pays strictly more. Reproducing it in mean is not enough for the criterion at hand: against the aligned attack with $a = \pm\Delta$ the realized peak is $1 + \Delta \mp \theta_1$, so a centered but non-degenerate sampler exceeds $\gamma^\star$ with positive probability. Sampling is randomized tie-breaking; the only ``posterior'' that is optimal path by path is the point mass at the midpoint.

\section{Regret cannot certify safety}\label{sec:regret}

Section~\ref{sec:exploration} priced the standard tools in the problem's own metric, the worst-case peak. The natural objection is that this is the wrong court: those mechanisms were designed to minimize \emph{regret} (cumulative cost minus that of a clairvoyant benchmark), the dominant performance metric of online learning-to-control. This section answers the objection. The next proposition shows that in this problem regret-style guarantees and worst-case safety are \emph{incomparable}: arbitrarily slowly growing regret is compatible with unbounded peak, and, perhaps more surprisingly, the minimax-optimal law itself has \emph{linear} regret under cumulative state cost. The certificates that recommend exploration can neither see its worst-case cost nor recognize the optimal policy.

\begin{proposition}[Regret is blind to the peak, in both directions]\label{prop:regret}
Fix $\varepsilon > 0$ and let the mistake cost be $c_t = \mathbf{1}\{\abs{x_t} > 1+\varepsilon\}$; the clairvoyant policy (which knows $a$) makes no mistakes ($x_0 = 0$ and deadbeat control keeps $\abs{x_t} \le 1$). Define $R_T(K; a,w) := \sum_{t\le T} c_t - \sum_{t \le T} c^*_t$.
\begin{enumerate}[label=(\roman*), nosep]
    \item \emph{(Small regret, infinite peak.)} For every sparse schedule $t_1 < t_2 < \cdots \uparrow \infty$ (meaning $t_{k+1} - t_k \ge 4$ for all but finitely many $k$; shorter gaps only enlarge $C(\Delta,\varepsilon)$) there is a causal controller $K^{(\mathrm{probe})}$ (the optimal law with isolated probes $p_k \uparrow \infty$ superimposed at times $t_k$) such that for \emph{every} admissible $(a,w)$: $R_T \le C(\Delta, \varepsilon) + 3\,\#\{k : t_k \le T\}$, where $C(\Delta,\varepsilon)$ counts the (at most $O(\log(\Delta^2/\varepsilon))$) mistakes before and during the first probe episode; this is $O(\log T)$ for $t_k = 2^k$ and $O(\log\log T)$ for $t_k = 2^{2^k}$. Yet $\sup_{a,w}\norm{x}_\infty = \infty$.
    \item \emph{(Optimal peak, linear regret.)} Under the cumulative state cost $c_t = \abs{x_t}$, the minimax-optimal law of Definition~\ref{def:controller} suffers, against the stealth pair $a = 1$, $w = \big(\tfrac{1}{\Delta+1}, 0, 0, \dots\big)$ (admissible for $\Delta \ge 1$; Appendix~\ref{app:regret} removes the restriction), regret exactly $R_T = \tfrac{1}{\Delta+1}(T - 1)$, linear in $T$: the attack holds $x_t \equiv \tfrac{1}{\Delta+1}$ forever while the induced data windows contain the entire prior, so no information ever leaks and the law never adapts, while the clairvoyant pays only the one-time disturbance.
\end{enumerate}
\end{proposition}

\begin{proof}[Proof sketch]
(i) Choose $p_1 \ge 2\Delta + 2 + 2(1+\Delta)/\varepsilon$. Before the first probe the controller is the plain optimal law, whose mistakes are bounded by $O(\log(\Delta^2/\varepsilon))$ uniformly in $(a,w)$: each state exceeding $1+\varepsilon$ shrinks the consistent interval by a constant factor (window lemma), and only $O(\log(\Delta^2/\varepsilon))$ such shrinkages fit between $2\Delta$ and the terminal width $\Theta(\varepsilon/\Delta)$. The first probe's own regressor then identifies $a$: the transition from $\abs{x} \ge p_1 - 1 - \Delta$ confines $a$ to a window of width $\le \varepsilon$, whatever $(a,w)$; thereafter the invariant \eqref{eq:invariant} (run from the post-probe consistent set) keeps $\abs{x_t} \le 1 + \varepsilon/2$ between probes, zero mistakes, while each later probe episode contributes at most $3$ mistakes. The peak, however, contains $\abs{x_{t_k+1}} \ge p_k \to \infty$ (e.g.\ under $(a,w) = (0,0)$). (ii) With $\abs{x_t} = \tfrac{1}{\Delta+1}$ and $\hat a_t = 0$, the data window has width $2\Delta + 2$ and is centered at $a = 1$, hence contains $[-\Delta, \Delta]$: $A_{t+1} = A_t$ and the loop reproduces itself. Full bookkeeping in Appendix~\ref{app:regret}.
\end{proof}

\begin{remark}[The incomparability is not an artifact of the accounting]\label{rem:accounting}
The two parts use different costs --- (i) a threshold cost, (ii) an integrated one --- but the incomparability does not rest on the switch: both directions hold under the single cost $c_t = \abs{x_t}$, (ii) as stated, and (i) with the geometric schedule $t_k = 2^{k-1}t_1$ and probe sizes $p_k = \Theta(\sqrt{t_k})$, for which $R_T = O(\sqrt T)$ \emph{uniformly over admissible $(a,w)$} while the worst-case peak remains infinite (bookkeeping and certified numerics in Appendix~\ref{app:regret}, part (iii)). The same holds under $x_t^2$ or $x_t^2 + u_t^2$, the standard integrated costs of the learning-for-control literature: along the stealth attack $u_t \equiv 0$, so every per-step cost that charges the held state is paid linearly (Appendix~\ref{app:regret}, part (ii), which also supplies the stealth pair for $\Delta < 1$). What the accounting does decide is which law regret \emph{prefers}: the comparison between the minimax-optimal law and the infinite-peak probing law points in opposite directions under the two accountings ($O(1)$ versus $\Theta(\#\text{probes})$ under the mistake cost, $\Theta(T)$ versus $O(\sqrt T)$ under the state cost), while the peak criterion ranks the pair unambiguously throughout (Figure~\ref{fig:probe}(c)). A certificate whose ranking of a fixed pair of policies on a fixed problem reverses under re-accounting of the same trajectories cannot carry safety semantics; and under the integrated yardstick it is the exactly minimax-optimal law that looks bad.
\end{remark}

The moral is not that regret is uninteresting, but that it is the wrong criterion for adversarial safety, in both directions. A time-average metric permits amortizing rare disasters, while a peak metric, the formal carrier of safety constraints, does not; and conversely a policy that is exactly optimal for the peak may look poor through a regret lens, because refusing to gamble on unidentified dynamics has a running cost that the clairvoyant does not pay. Any exploration scheme whose justification rests on amortization (explore-then-commit, $\varepsilon$-greedy schedules, optimism with logarithmic exploration budgets) imports a stochastic accounting into a setting where the adversary controls the exchange rate. This formalizes, at the level of a single solvable problem, the separation between mistake/safety guarantees and regret guarantees observed in \citet{ho2021online}, and complements the certainty-equivalence sensitivity analysis of \citet[Ch.~4]{ho2023thesis} and \citet{yu2023online}: in the closed-loop error dynamics of certainty-equivalent adaptation, the destabilizing gain is proportional to the \emph{movement} of the posited model between consecutive steps, not to its error, so a selector that keeps moving (as exploration forces it to) keeps the loop excited forever, while a competitive chaser moves only when data force it to, totaling at most $\Delta$ of movement over an infinite horizon (Proposition~\ref{prop:selector}(ii)).

\section{Numerical verification}\label{sec:numerics}

Independently of the proofs, Theorem~\ref{thm:main} was stress-tested numerically. For each $\Delta$ from $0.25$ to $500$, the closed loop of Definition~\ref{def:controller} was attacked by the theoretical worst case of Proposition~\ref{prop:lb}, by aligned and randomized disturbance strategies over parameter grids, by random-restart search over $(a, w)$, and by a \emph{lazy} (deferred-commitment) \emph{set-membership adversary} that at each step may realize any successor state consistent with some admissible $(a,w)$, an action space that provably contains every admissible adversary; every lazy rollout is converted into an explicit certificate (a parameter $a^\star$ in the final consistent set and reconstructed disturbances $w^\star_t$ verified to satisfy $\abs{w^\star_t}\le 1$) and replayed to confirm the identical trajectory. No attack ever exceeded $1 + \Delta$; the worst case of Proposition~\ref{prop:lb} attains it to machine precision at every $\Delta$; and independent exact reachability analysis of the scalar information state $(s_t, L_t)$ confirms the value with zero slack for $\Delta$ ranging from $0.125$ to $5\cdot10^4$. The forced-midpoint penalty of Proposition~\ref{prop:implicit}(iii) was checked the same way: a deviation $d$ at the post-kick history yields realized worst-case peak $1 + \Delta + d$ to $10^{-9}$, matching the exact linear penalty of Remark~\ref{rem:u0}. Figure~\ref{fig:verification} shows the zero-slack agreement and the realized worst-case trajectory.\footnote{A caveat for floating-point implementations: the $x_t = 0$ branch of the update \eqref{eq:setmem} is numerically delicate --- a state that is merely \emph{near} zero produces a window whose center divides rounding noise by $\abs{x_t}$, misplacing the posterior cut. Implementations should treat $\abs{x_t}$ below a threshold as zero; apparent bound violations in unguarded float simulations are artifacts of this knife edge and disappear under exact rational arithmetic.}

\begin{figure}[t]
    \centering
    \includegraphics[width=\textwidth]{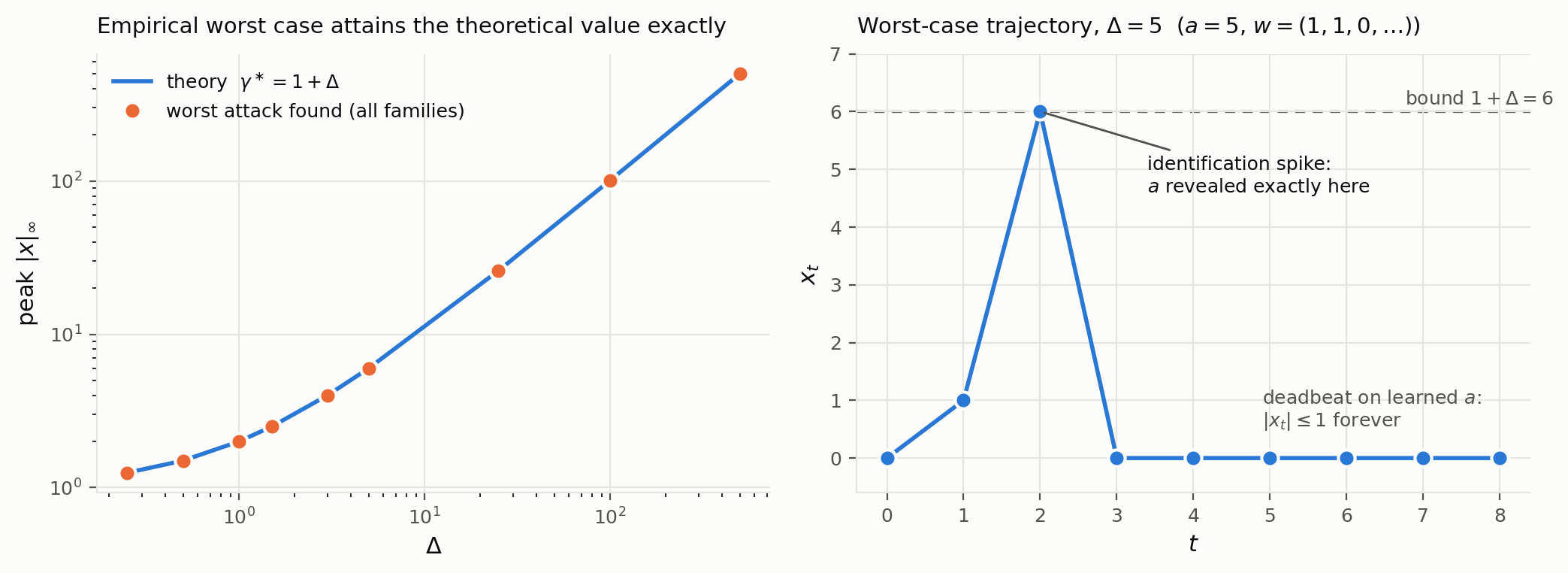}
    \caption{Empirical verification. \emph{Left:} the worst peak found by any attack family (markers) coincides with the theoretical value $\gamma^\star = 1 + \Delta$ (line) over more than three orders of magnitude of $\Delta$. \emph{Right:} the realized worst-case trajectory for $\Delta = 5$: the disturbance kicks the state to $1$ while revealing nothing about $a$; the aligned extreme parameter then produces the one-time identification spike $\abs{x_2} = \Delta + 1 = 6$, which reveals $a$ exactly, after which deadbeat control holds $\abs{x_t} \le 1$ forever.}
    \label{fig:verification}
\end{figure}

\section{Discussion and outlook}\label{sec:outlook}

\paragraph{What the exact solution teaches.} Three structural lessons emerge. First, \emph{the value is governed by the diameter of the uncertainty, not the instability of the nominal system} (Remark~\ref{rem:extensions}(c)): a known unstable pole is free; not knowing it is what costs. Second, \emph{constancy is what makes the price payable}: the same uncertainty set with a time-varying pole $a_t$ has infinite value for $\Delta \ge 1$ (Section~\ref{sec:related}), no observation ever constraining the next step, while the constant, hence identifiable, pole costs exactly $\Delta$; adaptation at large uncertainty works because the unknown holds still long enough to be learned. Third, \emph{the optimal way to learn adversarially is passively}: consistency extracts every bit the adversary's attack necessarily leaks, and competitiveness ensures the controller's model hypothesis moves no more than the data compel, the two properties that define consistent model chasing \citep{ho2021online}. The distinction is as old as this line of work: its first installment \citep{ho2018passive} introduced controllers that learn \emph{passively} near the origin and turn \emph{aggressive}, actively exciting the system, when pushed far from it; the present paper shows that under the peak objective the aggressive mode is priced out from rest and, at the optimum, only the passive mode survives.

\paragraph{The multivariate frontier.} The scalar solution is the ground floor of a tower. In dimension $n$ (unknown matrix $A$ with entrywise bound $\Delta$), our preliminary analysis supports the conjecture that the adversary can spend roughly one identification spike per direction, exactly one in the scalar case, so that the minimax peak grows as $\Delta^n$ up to $n$-dependent constants, with the natural generalization of Definition~\ref{def:controller} (per-coordinate functional centering on the consistent row sets) optimal at the level of this growth order. Exactness, however, will not come for free: the game is richer and the analysis harder --- already in $\R^2$ the adversary can profit from \emph{sacrificing} immediate peak to preserve uncertainty for a later, larger spike --- and the exact value is open. A structural reason makes the scalar case silent on what the answer will look like: in dimension one every notion of centrality collapses. Midpoint, Steiner point, Chebyshev center coincide, so the scalar certificate cannot reveal \emph{which} of them optimality selects; in higher dimensions they part ways, and identifying the right centering (the right consistent model chaser, the slot the architecture deliberately leaves open) is, we believe, where the exact theory will be decided. A complete account is the subject of ongoing work.

\paragraph{Nonzero initial conditions.} The game \eqref{eq:minimax} starts at rest, and the lower bound exploits precisely that: the kick buys the adversary a unit regressor at zero informational price. A nonzero initial state hands the adversary that regressor before play begins, and Corollary~\ref{cor:nonzeroic} settles every initial condition at or above the disturbance level --- the value is $\max(\abs{x_0},\, 1 + \Delta\abs{x_0})$, the midpoint law attains it, and outside the degenerate case the midpoint is again forced at the start --- with the same window lemma and invariant carrying the proof. What remains open is the sub-disturbance window $0 < \abs{x_0} < 1$, where the lower bound $1 + \Delta\abs{x_0}$ parts ways with the law's guarantee of $1 + \Delta$. Safety is not in question there, and the window's endpoints $\abs{x_0} = 0$ and $\abs{x_0} = 1$ both carry the value $1 + \Delta$; what is open is the exact value between them, which we defer.

\paragraph{The selector is the load-bearing component: a data-first case study.} The optimal architecture admits an instructive comparison with the model-free cancellation controller of \citet{ho2019robust} (see also \citealp[Ch.~5]{ho2023thesis}), which stabilizes \eqref{eq:sys} (in any dimension) with \emph{no} prior parameter bound, directly from data matrices. Specialized to $n=1$, that controller is itself a certainty-equivalent deadbeat law, sharing the oracle and the passive, consistency-driven structure of Definition~\ref{def:controller}, but its selector retains \emph{one} data window, the one generated by the largest past regressor (a fixed initialization entry $\varepsilon > 0$ acting as a phantom regressor positing $\hat a = 0$), in place of the full consistent set $A_t$. That controller, however, solves a harder problem than \eqref{eq:minimax}: it assumes no bound on the parameter and no knowledge of the disturbance bound, and part of the gap below is the price of that generality. The difference is still not cosmetic. Kicking the state to $\abs{x_1} = c \in (\varepsilon, 1]$ and cashing in immediately forces a peak of $\Delta + 1/c + 1$, unbounded as $\varepsilon \to 0$, and hiding the first state below $\varepsilon$ forces, at the natural scale $\varepsilon = 1$, a peak of $\Delta^2 + \Delta + 1$: one full extra multiplicative stage.\footnote{Both formulas are realized by explicit three-move attacks, replayed numerically across $\Delta \in [1.5, 25]$. Cash-in: $w_0 = c$ moves the state to $c$ while the phantom regressor still governs, so $u_1 = 0$; the aligned pair $a = \Delta$, $w_1 = 1$ gives $x_2 = \Delta c + 1$, measured through the regressor $c$, so the retained window posits $\hat a = \Delta + 1/c$, off by $1/c$; the aligned $w_2$ then realizes $\abs{x_3} = \tfrac1c(\Delta c + 1) + 1 = \Delta + 1/c + 1$. Masking, at $\varepsilon = 1$: $w_0 = \delta < 1$ keeps the phantom regressor in charge for two steps ($u_1 = u_2 = 0$), so the aligned $a = \Delta$ compounds twice, $\abs{x_3} = \Delta^2\delta + \Delta + 1 \uparrow \Delta^2 + \Delta + 1$ as $\delta \uparrow 1$.} Repaired variants fall short as well: clipping the estimate to the prior still pays order $2\Delta$ under the same masking attack, and intersecting \emph{all} data windows without the prior does not help, because at the critical spike moment only one window exists and it is the \emph{prior} that caps the estimation error at $\Delta$. The exact optimum $1 + \Delta$ is restored precisely when the selector keeps the entire consistent set, all windows \emph{and} the prior, and centers on it; the resulting law remains purely data-first, the endpoints of $A_t$ being extrema of the raw-data functionals $(x_{j+1} - u_j \pm 1)/x_j$ and $\pm\Delta$. The comparison sharpens the conclusion of Section~\ref{sec:architecture}: with the oracle fixed, the entire distance between ``stabilizing'' and ``minimax optimal'' is carried by the selector.

\paragraph{Toward a model-chasing theory of adaptive control.} The classical stochastic theory of adaptive control has certainty equivalence and persistent excitation as its organizing principles. The results here, combined with the framework of \citet{ho2021online} and its sequels, suggest an organizing principle for the adversarial regime: \emph{robust oracle $\times$ consistent model chasing}, with optimality certificates in place of asymptotic ones. The scalar problem now has a complete such certificate; extending exactness to richer uncertainty classes, partial observation, and input constraints are natural next steps.

\bibliographystyle{plainnat}
\bibliography{refs}

\appendix

\section{The framework interface conditions: statements and proofs}\label{app:interface}

This appendix states the two interface conditions of the oracle--selector framework, specialized to our setting --- they are formulated for general nonlinear systems and compact parameter spaces in \citet{ho2021online}, and developed in full in \citet[Ch.~6]{ho2023thesis} --- and proves that the two components of Definition~\ref{def:controller} satisfy them, as claimed in Section~\ref{sec:architecture}.

\begin{definition}[Robust oracle; consistent model chaser \citep{ho2021online}]\label{def:interface}
Let $(\Theta, d)$ be a compact parameter space of models. \emph{(i)} A map $\pi$ from parameters to control policies is a \emph{$\rho$-robust oracle} for an objective encoded by binary costs $\mathcal{G}_t(x_t, u_t) \in \{0,1\}$ (``mistakes'') if, whenever the policies $\pi[\theta_t]$ are applied with any parameter sequence satisfying $d(\theta_t, \theta^*) \le \rho$ for the true $\theta^*$, the total number of mistakes is finite; the \emph{mistake function} $M^\pi_\rho(\gamma)$ is the worst-case number of mistakes from initial states of norm at most $\gamma$. \emph{(ii)} A map $\Sel$ from data sets to parameters is a \emph{consistent model chaser} if it always selects a parameter consistent with all data seen so far, and it is \emph{$\gamma_{\mathrm{c}}$-competitive} if over every time window its selections' total movement is bounded by $\gamma_{\mathrm{c}}$ times the Hausdorff distance between the consistent sets at the window's ends: $\sum_{t = t_1+1}^{t_2} d(\theta_t, \theta_{t-1}) \le \gamma_{\mathrm{c}}\, \dH\big(P_{t_1}, P_{t_2}\big)$, where $P_t$ is the consistent set at time $t$. The composition $\Api$ plays, at every step, $u_t = \pi[\Sel(\text{data}_t)](x_t)$.
\end{definition}

If $\pi$ is a robust oracle and $\Sel$ a competitive chaser, the general theorems of \citet{ho2021online} guarantee that the composition $\Api$ inherits worst-case boundedness and mistake bounds \emph{for arbitrarily large parameter sets}.

\begin{proposition}[Deadbeat is a robust oracle]\label{prop:oracle}
Let $\pi[\theta](x) := -\theta x$ for $\theta \in \R$. For any $\rho \in (0,1)$ and any parameter sequence $(\theta_t)$ with $\abs{\theta_t - a} \le \rho$ for all $t$, the closed loop $x_{t+1} = a x_t + \pi[\theta_t](x_t) + w_t$ satisfies
\[
    \abs{x_{t+1}} \le \rho\abs{x_t} + 1,
    \qquad\text{hence}\qquad
    \abs{x_t} \le \tfrac{1}{1-\rho} + \rho^t\big(\abs{x_0} - \tfrac{1}{1-\rho}\big),
\]
i.e., $\pi$ is a $\rho$-robust control oracle in the sense of \citet{ho2021online} for every interval-membership objective $\mathbf{1}\{\abs{x_t} > b\}$ with threshold $b > \tfrac{1}{1-\rho}$ (the displayed geometric envelope is the corresponding peak guarantee), with mistake function $M^\pi_\rho(\gamma) \le \log_+\gamma/\log\rho^{-1} + c(\rho, b)$, where $\log_+ := \max(\log, 0)$.
\end{proposition}

\begin{proof}[Proof of Proposition~\ref{prop:oracle}]
With $\abs{\theta_t - a} \le \rho$ the closed loop reads $x_{t+1} = (a - \theta_t)x_t + w_t$, so $\abs{x_{t+1}} \le \rho\abs{x_t} + 1$; iterating gives the displayed envelope. For the mistake bound set $c_0 := \tfrac{1}{1-\rho}$ and $y_t := \abs{x_t} - c_0$, so that $y_{t+1} \le \rho\,y_t$ and $y_0 \le \gamma - c_0$. A mistake at time $t$ means $\abs{x_t} > b$, i.e.\ $y_t > b - c_0$, which is positive since $b > c_0$; as $y_t \le \rho^t(\gamma - c_0)$ whenever $\gamma > c_0$ (and $y_t \le 0$ forever otherwise), a mistake requires $\rho^t(\gamma - c_0) > b - c_0$, i.e.\ $t < \log\big((\gamma - c_0)/(b - c_0)\big)/\log\rho^{-1}$. Counting the admissible $t$,
\[
    M^\pi_\rho(\gamma) \;\le\; 1 + \frac{\log_+\gamma + \log_+\tfrac{1}{b - c_0}}{\log\rho^{-1}}
    \;=\; \frac{\log_+\gamma}{\log\rho^{-1}} + c(\rho, b),
    \qquad c(\rho, b) := 1 + \frac{\log_+\tfrac{1}{b - c_0}}{\log\rho^{-1}}.
\]
The threshold condition is sharp: at $b = c_0$, from any $\gamma > c_0$ the adversary can realize $\abs{x_t} = c_0 + \rho^t(\gamma - c_0) > b$ at every $t$, so no finite mistake bound exists for $b \le c_0$.
\end{proof}

\begin{proposition}[Midpoint is a $1$-competitive consistent model chaser]\label{prop:selector}
Let $(A_t)$ be any nested sequence of closed intervals and $m_t := \midpt(A_t)$. Then:
\begin{enumerate}[label=(\roman*), nosep]
    \item $m_t \in A_t$ for all $t$ (consistency);
    \item $\abs{m_{t+1} - m_t} \le \tfrac12\big(L_t - L_{t+1}\big)$, so the total movement satisfies $\sum_{t \ge 0}\abs{m_{t+1}-m_t} \le \tfrac12\big(L_0 - \lim_t L_t\big) \le L_0/2$ ($= \Delta$ for the sets of Definition~\ref{def:controller});
    \item for all $t_1 < t_2$: $\sum_{t=t_1+1}^{t_2} \abs{m_t - m_{t-1}} \le \dH(A_{t_1}, A_{t_2})$, i.e., $\midpt$ is a $1$-competitive consistent model chaser in the sense of Definition~\ref{def:interface}.
\end{enumerate}
Moreover, $\midpt(A_t)$ coincides with the Steiner point of $A_t$ in dimension one, so this selector is the one-dimensional instance of the Steiner-point chasing algorithm of \citet{bubeck2020chasing} employed in \citet{ho2021online}.
\end{proposition}

\begin{proof}[Proof of Proposition~\ref{prop:selector}]
(i) is immediate. For (ii), write $A_t = [l, r] \supseteq A_{t+1} = [l', r']$ and set $g_l := l'-l \ge 0$, $g_r := r - r' \ge 0$; then $m_{t+1} - m_t = (g_l - g_r)/2$, so $\abs{m_{t+1}-m_t} \le (g_l + g_r)/2 = (L_t - L_{t+1})/2$; telescoping gives the total bound. For (iii), summing the per-step bound over $(t_1, t_2]$ gives $\tfrac12\big(G_l + G_r\big)$ where $G_l, G_r \ge 0$ are the total endpoint displacements between $A_{t_1}$ and $A_{t_2}$; since for nested intervals $\dH(A_{t_1}, A_{t_2}) = \max(G_l, G_r) \ge \tfrac12(G_l + G_r)$, the claim follows.
\end{proof}

\section{Attack primitives and lower-bound witnesses}\label{app:echo}

Throughout, $s_t := \abs{x_t}$, $L_t := \abs{A_t}$, and $E_t := \sup_{\alpha \in A_t}\abs{\alpha - \theta_t}$ (the implicit selector's worst-case error, as in Section~\ref{sec:ub}), with $h$ as in Proposition~\ref{prop:echo}. This appendix names the three attack moves used by the lower bounds of Section~\ref{sec:exploration} and runs them, at every error cap, to produce the witnesses invoked by Proposition~\ref{prop:echo}.

\paragraph{The attack vocabulary.}
The lower bounds of Section~\ref{sec:exploration} are assembled from three moves, which we name once here and refer to by name thereafter. Throughout, the controller is a certainty-equivalent deadbeat law $u_t = -\theta_t x_t$ with an arbitrary consistent selector --- equivalently, by Definition~\ref{def:implicit}, any causal law whose implicit selector is consistent --- and $s_t$, $L_t$, $E_t$ are as above.

\begin{definition}[Attack primitives]\label{def:attacks}
\leavevmode
\begin{enumerate}[label=(\alph*), nosep, leftmargin=*]
\item \emph{Kick.} At $t = 0$ the state $x_0 = 0$ forces $u_0 = 0$, so $w_0 = \pm 1$ realizes $s_1 = 1$; and because the transition's regressor is $x_0 = 0$, the induced data window is all of $\R$ and $A_1 = A_0$. A unit regressor, bought at zero information.
\item \emph{Spike of height $y$.} From $s_t = 1$ with selected point $\theta_t$, let $c$ be the corner of $A_t$ farthest from $\theta_t$, so $\abs{c - \theta_t} = E_t$. The adversary drives the state toward $c$ but \emph{holds back} by $E_t - y$, realizing $s_{t+1} = y$; admissibility of $\abs{w_t} \le 1$ makes exactly $y \in [E_t - 1,\, E_t + 1]$ available. The induced width-$2$ data window then sits at the $c$-end of $A_t$, and for $y \ge 1$ the surviving posterior abuts $c$ with length
\[
    \mu \;=\; \min\big(2,\; E_t + 1 - y\big).
\]
Raising $y$ slides the window off the far corner and shortens $\mu$ one-for-one. This is the additive tradeoff of Lemma~\ref{lem:window}(i) read as a design knob.
\item \emph{Echo.} The adversary does \emph{not} commit to $a$ at the spike. It waits until $\theta_{t+1}$ is revealed, takes $a$ to be the endpoint of $A_{t+1}$ opposite $\theta_{t+1}$ (at distance $E_{t+1}$, whichever endpoint the selector chose), and aligns $w_{t+1}$, producing $s_{t+2} = E_{t+1}\,y + 1$. Deferral is legitimate by the certificate argument in the proof of Proposition~\ref{prop:implicit}(i): the required certificate is $a \in A_{t+1}$ together with reconstructed disturbances of magnitude $\le 1$, and the data window is by construction exactly the set of parameters meeting the latter.
\end{enumerate}
\end{definition}

The \emph{half-spike} of Proposition~\ref{prop:echo} and the \emph{vertex-spike} introduced here are not two attacks but the two optimal tunings of the single knob $y$ in their respective regimes (the \emph{full-spike} of Table~\ref{tab:witness} is the aligned kick of Proposition~\ref{prop:endpoint}, needing no echo): the echo payoff is $\mu\,y = y\,\min(2,\, E_t + 1 - y)$, and maximizing it over the feasible range is a one-line constrained maximization whose two branches are the two branches of $h$.

\begin{figure}[t]
\centering
\begin{tikzpicture}
  \begin{scope}
    \fill[black!14] (-2.5,-0.16) rectangle (-0.5,0.16);
    \draw[very thick] (-2.5,0) -- (2.5,0);
    \draw[very thick] (-2.5,-0.24) -- (-2.5,0.24);
    \draw[very thick] (2.5,-0.24) -- (2.5,0.24);
    \draw (-0.5,-0.24) -- (-0.5,0.24);
    \node[font=\small, anchor=south] at (-2.5,0.30) {$a$};
    \node[font=\small, anchor=south] at (2.5,0.30) {$\theta_1$};
    \draw[gray, dashed] (-2.5,0.62) -- (-2.5,1.05);
    \draw[gray, dashed] (-0.5,0.30) -- (-0.5,1.05);
    \draw[{Latex}-{Latex}] (-2.5,1.05) -- (-0.5,1.05);
    \node[font=\small, anchor=south] at (-1.5,1.10) {data window, width $2$};
    \draw[{Latex}-{Latex}] (-2.5,-0.62) -- (-0.5,-0.62);
    \node[font=\small, anchor=north] at (-1.5,-0.68) {$\mu = 2$};
    \node[font=\small, anchor=north] at (0,-1.45) {\textbf{half-spike:} $y = \ecap - 1$};
    \node[font=\small, anchor=north] at (0,-1.90) {window just fits; $\mu$ is maximal};
    \node[font=\footnotesize, anchor=north] at (0,-2.35) {($\ecap = 5$, $y = 4$)};
  \end{scope}
  \begin{scope}[xshift=8.2cm]
    \fill[black!14] (-2.5,-0.16) rectangle (1.0,0.16);
    \draw[very thick] (-2.5,0) -- (2.5,0);
    \draw[very thick] (-2.5,-0.24) -- (-2.5,0.24);
    \draw[very thick] (2.5,-0.24) -- (2.5,0.24);
    \draw (1.0,-0.24) -- (1.0,0.24);
    \draw[thick, dotted] (-3.0,-0.16) rectangle (1.0,0.16);
    \node[font=\small, anchor=south] at (-2.5,0.30) {$a$};
    \node[font=\small, anchor=south] at (2.5,0.30) {$\theta_1$};
    \draw[gray, dashed] (-3.0,0.16) -- (-3.0,1.05);
    \draw[gray, dashed] (1.0,0.30) -- (1.0,1.05);
    \draw[{Latex}-{Latex}] (-3.0,1.05) -- (1.0,1.05);
    \node[font=\small, anchor=south] at (-1.0,1.10) {data window, width $2$};
    \draw[{Latex}-{Latex}] (-2.5,-0.62) -- (1.0,-0.62);
    \node[font=\small, anchor=north] at (-0.75,-0.68) {$\mu = \ecap + 1 - y$};
    \node[font=\small, anchor=north] at (-0.25,-1.45) {\textbf{vertex-spike:} $y = \tfrac{\ecap+1}{2}$};
    \node[font=\small, anchor=north] at (-0.25,-1.90) {window overhangs; $\mu$ shrinks};
    \node[font=\footnotesize, anchor=north] at (-0.25,-2.35) {($\ecap = 5/2$, $y = 7/4$)};
  \end{scope}
\end{tikzpicture}
\caption{The spike knob, in the regressor coordinate of Corollary~\ref{cor:geom}. Thick segment: the prior $A_1 = [-\tfrac{\ecap}{2}, \tfrac{\ecap}{2}]$ after the kick, with the selected endpoint $\theta_1$ at one end. Shaded: the surviving posterior $A_2$, which always abuts the far corner $a$. Raising the spike height $y$ slides the width-$2$ data window leftward off that corner. \emph{Left:} at $y = \ecap-1$ the window just fits inside the prior and the full length $\mu = 2$ survives. \emph{Right:} at $y > \ecap-1$ the window overhangs and $\mu = \ecap+1-y$ shrinks one-for-one --- the tradeoff the adversary optimizes.}
\label{fig:spike}
\end{figure}
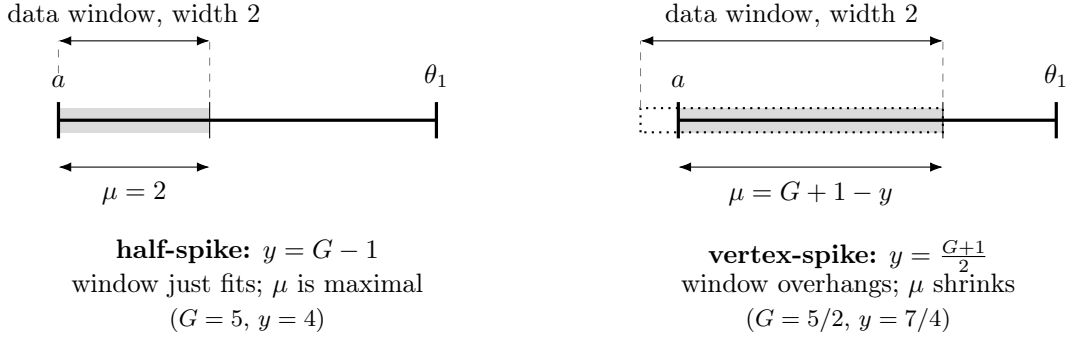

\paragraph{The witnesses.}
The next lemma runs the three moves against the endpoint law at every error cap $\ecap$ and tunes the spike optimally; instantiated at $\ecap = 2\Delta$ it yields the full piecewise lower bound of Proposition~\ref{prop:echo}. Table~\ref{tab:run} replays the run step by step --- each row is one time step, carrying every quantity the proof references at that step.

\begin{lemma}[The witness family: $1 + h(\ecap)$ at every error cap]\label{lem:lb}
Fix $\ecap > 0$, run the game on the prior $A_0 = [-\ecap/2,\, \ecap/2]$, and let $K^{\mathrm{end}}$ be the certainty-equivalent deadbeat law selecting an endpoint of $A_t$ by any rule. Then $E_t = L_t \le \ecap$ at every step, and there is an admissible pair $(a, w)$ with
\[
    \norm{x}_\infty \;=\; 1 + h(\ecap) \quad (\ecap \ge 1), \qquad\qquad \norm{x}_\infty \;=\; 1 + \ecap \quad (\ecap \le 1).
\]
The step-by-step run is that of Table~\ref{tab:run}; the optimal spike heights, and the resulting trajectories, are those of Table~\ref{tab:witness}.
\end{lemma}

\begin{proof}
\emph{Kick.} $w_0 = 1$ gives $x_1 = 1$ and $A_1 = [-\ecap/2, \ecap/2]$, so $L_1 = E_1 = \ecap$. Let $\theta_1 \in \{\pm \ecap/2\}$ be the endpoint the rule selects, $c := -\theta_1$ the opposite corner, and mirror if necessary so that $\theta_1 = \ecap/2$, $c = -\ecap/2$.

\emph{Spike of height $y$.} Put $x_2 := -y$. Since $u_1 = -\theta_1 x_1 = -\ecap/2$, the disturbance needed to realize this from a parameter $\alpha$ is $w_1 = x_2 - \alpha x_1 - u_1 = \ecap/2 - y - \alpha$, so the set of parameters consistent with the transition is exactly the window $[\,\ecap/2 - y - 1,\; \ecap/2 - y + 1\,]$, and every $\alpha$ in it certifies $\abs{w_1} \le 1$. Feasibility of the spike is $c$ lying in the window, i.e.\ $\ecap - 1 \le y \le \ecap + 1$ (the window meets $A_1$ at all iff $y \le \ecap+1$). For $\max(1,\, \ecap-1) \le y \le \ecap+1$ the window's upper end $\ecap/2 - y + 1$ lies at or below $\ecap/2$ and its lower end at or below $-\ecap/2$, so
\[
    A_2 \;=\; [\,-\ecap/2,\; \min(\ecap/2,\; \ecap/2 - y + 1)\,], \qquad \mu \;:=\; L_2 \;=\; \min\big(2,\; \ecap + 1 - y\big),
\]
an interval abutting $c$, as in Definition~\ref{def:attacks}(b) and Figure~\ref{fig:spike}.

\emph{Echo.} The selector reveals $\theta_2 \in \partial A_2$; the adversary sets $a$ to the opposite endpoint of $A_2$, so $\abs{a - \theta_2} = \mu = E_2$, and $a \in A_2$ certifies $\abs{w_1} \le 1$ by the previous paragraph. Aligning $w_2 := \sgn\big((a - \theta_2)x_2\big)$ gives $\abs{x_3} = \mu\,y + 1$, and the realized peak is $\max(1,\, y,\, \mu y + 1) = \mu y + 1$.

\emph{Optimizing the knob.} It remains to maximize $\mu y = y\,\min(2,\, \ecap+1-y)$ over $y \ge \max(1,\, \ecap-1)$, $y \le \ecap+1$. The unconstrained vertex of $y(\ecap+1-y)$ is $y = \tfrac{\ecap+1}{2}$ with value $\big(\tfrac{\ecap+1}{2}\big)^2$; it respects $\mu \le 2$ iff $\tfrac{\ecap+1}{2} \le 2$, i.e.\ \emph{iff} $\ecap \le 3$, and respects $y \ge 1$ \emph{iff} $\ecap \ge 1$. Hence:
\begin{itemize}[nosep]
\item \emph{$\ecap \ge 3$:} the feasible range is $y \in [\ecap-1,\, \ecap+1]$, on which the objective $y(\ecap+1-y)$ is decreasing (the vertex lies below $\ecap-1$), so the optimum is at $y = \ecap-1$, the half-spike, where the window just fits and $\mu = 2$, giving $\mu y = 2(\ecap-1) = h(\ecap)$ and peak $2\ecap - 1$.
\item \emph{$1 \le \ecap \le 3$:} the vertex is admissible --- the vertex-spike --- giving $\mu y = \big(\tfrac{\ecap+1}{2}\big)^2 = h(\ecap)$ and peak $1 + h(\ecap)$.
\item \emph{$\ecap \le 1$:} the constraint $y \ge 1$ binds, so the echo yields at most $\mu y \le \ecap$ at $y = 1$; the same value is reached directly by the \emph{full-spike} $y = \ecap+1$, which is the aligned kick of Proposition~\ref{prop:endpoint} and gives peak $\abs{x_2} = 1 + \ecap$ with no echo needed.
\end{itemize}
In each case the realized peak is the tabulated $1 + h(\ecap)$ for $\ecap \ge 1$, respectively $1 + \ecap$ for $\ecap \le 1$.
\end{proof}

\begin{table}[t]
\centering
\footnotesize
\setlength{\tabcolsep}{3.5pt}
\begin{tabular}{@{}cl c cc cc cc c@{}}
\toprule
&& state & \multicolumn{2}{c}{consistent set} & \multicolumn{2}{c}{controller} & \multicolumn{2}{c}{adversary} & transition \\
\cmidrule(lr){3-3}\cmidrule(lr){4-5}\cmidrule(lr){6-7}\cmidrule(lr){8-9}\cmidrule(l){10-10}
$t$ & move & $x_t$ & $A_t$ & $L_t{=}E_t$ & $\theta_t$ & $u_t$ & $a$ & $w_t$ & window $W_t$ \\
\midrule
$0$ & kick  & $0$  & $[-\tfrac{\ecap}{2},\tfrac{\ecap}{2}]$ & $\ecap$ & --- & $0$ & free & $+1$ & $\R$ \\
$1$ & spike & $1$  & $[-\tfrac{\ecap}{2},\tfrac{\ecap}{2}]$ & $\ecap$ & $\tfrac{\ecap}{2}$ & $-\tfrac{\ecap}{2}$ & free & $\tfrac{\ecap}{2}{-}y{-}a$ & $[\tfrac{\ecap}{2}{-}y{-}1,\;\tfrac{\ecap}{2}{-}y{+}1]$ \\
$2$ & echo  & $-y$ & $[-\tfrac{\ecap}{2},\;\tfrac{\ecap}{2}{-}y{+}1]$ & $\mu$ & ${\in}\;\partial A_2$ & $\theta_2\,y$ & opp.\ end of $A_2$ & $\sgn\!\big((a{-}\theta_2)x_2\big)$ & --- \\
$3$ & ---   & $\pm(\mu y{+}1)$ & \multicolumn{7}{l}{peak $= \mu y + 1$, \ where $\mu = \min(2,\; \ecap{+}1{-}y)$} \\
\bottomrule
\end{tabular}
\caption{The witness run of Lemma~\ref{lem:lb} at generic spike height $y$, on the prior $[-\tfrac{\ecap}{2}, \tfrac{\ecap}{2}]$, mirrored so that $\theta_1 = \tfrac{\ecap}{2}$ (far corner $c = -\tfrac{\ecap}{2}$). Row $t$ is the game state at step $t$: $A_t$ is the consistent set \emph{entering} the step, $W_t$ the data window induced by the transition out of it ($A_{t+1} = A_t \cap W_t$), and $E_t = L_t$ because the selector picks endpoints. The adversary defers committing to $a$ until $\theta_2$ is revealed: any $a \in A_2$ certifies $\abs{w_1} \le 1$ through the tabulated reconstruction $w_1 = \tfrac{\ecap}{2} - y - a$. For $\ecap \le 1$ the optimum is the full-spike $y = \ecap+1$, ending at $t = 2$ with no echo. Optimal tunings of $y$: Table~\ref{tab:witness}.}
\label{tab:run}
\end{table}

\begin{table}[t]
\centering
\small
\begin{tabular}{lllll}
\toprule
regime & spike height $y$ & posterior $\mu$ & trajectory $(x_1, x_2, x_3)$ & peak \\
\midrule
$\ecap \le 1$ & $\ecap + 1$ (full-spike) & --- (no echo) & $(1,\; -(1+\ecap))$ & $1 + \ecap$ \\
$1 \le \ecap \le 3$ & $\tfrac{\ecap+1}{2}$ (vertex) & $\tfrac{\ecap+1}{2}$ & $\big(1,\; -\tfrac{\ecap+1}{2},\; \pm(1 + h(\ecap))\big)$ & $1 + \big(\tfrac{\ecap+1}{2}\big)^2$ \\
$\ecap \ge 3$ & $\ecap - 1$ (half-spike) & $2$ & $\big(1,\; -(\ecap-1),\; \pm(2\ecap-1)\big)$ & $2\ecap - 1$ \\
\bottomrule
\end{tabular}
\caption{The witnesses of Lemma~\ref{lem:lb}, on the prior $[-\tfrac{\ecap}{2}, \tfrac{\ecap}{2}]$ against the endpoint law. The sign of $x_3$ is whichever the selector's tie-break forces; the adversary commits to $a$ only after seeing $\theta_2$. Both branch points of $h$ are feasibility boundaries of the single spike knob: $\ecap = 3$ is where the width-$2$ window stops fitting inside the prior ($\mu \le 2$ starts to bind), and $\ecap = 1$ is where the vertex height drops below $1$ and the echo dies. Certified numerically for $\ecap$ across three orders of magnitude ($0.05$ to $200$) and for both endpoint tie-breaks, each attack replayed through the closed loop from its reconstructed certificate.}
\label{tab:witness}
\end{table}

\section{Proof details for Proposition~\ref{prop:regret}}\label{app:regret}

\paragraph{(i) Small regret, infinite peak.}
The schedule, the anatomy of one probe episode, and the resulting ranking flip are pictured in Figure~\ref{fig:probe}. Fix $\varepsilon > 0$, a sparse schedule $t_1 < t_2 < \cdots \uparrow \infty$, and probe sizes $p_k \uparrow \infty$ with $p_1 \ge 2\Delta + 2 + 2(1+\Delta)/\varepsilon$. The controller is $u_t = -\midpt(A_t)\,x_t + p_k\,\mathbf{1}\{t = t_k\}$. Fix any admissible $(a, w)$. Table~\ref{tab:episodes} tabulates, step by step, the bounds derived in the next three paragraphs.

\emph{Before the first probe} the controller coincides with the optimal law, which keeps $\abs{x_t} \le 1 + \Delta$; every transition that lands above the threshold, $\abs{x_{t+1}} > 1 + \varepsilon$, has, by Corollary~\ref{cor:geom}(i), posterior length $s_t L_{t+1} \le g_t/2 + 1 - \abs{x_{t+1}} < s_t L_t/2 - \varepsilon$, so each mistake more than halves the surviving interval length; consequently at most $C(\Delta, \varepsilon) = O(\log(\Delta^2/\varepsilon))$ mistakes occur in this phase, uniformly over $(a,w)$.

\emph{First probe episode.} The probe fires from $\abs{x_{t_1}} \le 1 + \Delta$, so the probe state satisfies $\abs{x_{t_1+1}} \le p_1 + 1 + \Delta$ and its regressor is at least $p_1 - 1 - \Delta$; the following step is bounded by the worst exposure, $\abs{x_{t_1+2}} \le \Delta\big(p_1+1+\Delta\big) + 1$. The transition \emph{from} the probe state, however, confines $a$ to a window of width $L \le 2/(p_1 - 1 - \Delta) \le \varepsilon/(1 + \Delta)$, whatever $(a,w)$. The echo is then bounded, $\abs{x_{t_1+3}} \le \tfrac{L}{2}\abs{x_{t_1+2}} + 1 \le \Delta(1+\varepsilon) + O(1)$ (one further $\Delta$-scale mistake), and since $t_1+2$ is a plain-law step, the invariant \eqref{eq:invariant} applied there gives $g_{t_1+3} \le L_{t_1+2} \le L$, whence $\abs{x_{t_1+4}} \le L/2 + 1 \le 1 + \varepsilon/2$: the episode contributes $O(1)$ mistakes, absorbed into $C(\Delta,\varepsilon)$.

\emph{Later probes.} From $\abs{x} \le 1 + \varepsilon$ with $L \le \varepsilon/(1+\Delta)$: the probe state is at most $p_k + 2$ (one mistake); the next state is at most $\tfrac{L}{2}(p_k+2) + 1$ (a second, possibly large, mistake); but the probe-$k$ transition itself re-identifies $a$ to width $\le 2/(p_k - 2)$, so the step after satisfies $\abs{x} \le \tfrac{1}{p_k-2}\big(\tfrac{L}{2}(p_k+2)+1\big) + 1 \le 1 + \varepsilon$ for $p_k$ large: at most $3$ mistakes per episode, uniformly in $k$.

\emph{Between probes.} Proposition~\ref{prop:invariant}, applied from the post-episode state, gives exposure $g_t \le L(1+\varepsilon) \le \varepsilon$ (for $\varepsilon \le \Delta$; larger $\varepsilon$ only weakens the mistake criterion), hence $\abs{x_t} \le 1 + \varepsilon/2$: zero mistakes.

Summing: $R_T \le C(\Delta,\varepsilon) + 3\,\#\{k : t_k \le T\}$ for \emph{every} admissible $(a,w)$; with $t_k = 2^k$ this is $O(\log T)$, with $t_k = 2^{2^k}$ it is $O(\log\log T)$. The worst-case peak is nevertheless infinite: under the admissible pair $(a, w) = (0, 0)$ the trajectory contains $x_{t_k+1} = p_k \to \infty$.

\begin{table}[t]
\centering
\footnotesize
\setlength{\tabcolsep}{4pt}
\begin{tabular}{@{}l l l c@{}}
\toprule
$t$ & bound on $\abs{x_t}$ & information state after the step & mistake? \\
\midrule
\multicolumn{4}{@{}l}{\emph{First episode} --- entered from the plain law: $\abs{x_{t_1}} \le 1+\Delta$, \ $L_{t_1} \le 2\Delta$.}\\
$t_1$     & $1+\Delta$ & --- & (in $C$) \\
$t_1{+}1$ & $p_1{+}1{+}\Delta$ \ (probe) & regressor ${\ge}\,p_1{-}1{-}\Delta$: \ $L \le \tfrac{2}{p_1-1-\Delta} \le \tfrac{\varepsilon}{1+\Delta}$ & yes \\
$t_1{+}2$ & $\Delta(p_1{+}1{+}\Delta)+1$ \ (exposure) & plain-law step: \eqref{eq:invariant} gives $g_{t_1+3} \le L_{t_1+2} \le L$ & yes \\
$t_1{+}3$ & $\tfrac{L}{2}\abs{x_{t_1+2}}+1 \le \Delta(1{+}\varepsilon)+O(1)$ \ (echo) & --- & $\le 1$ \\
$t_1{+}4$ & $L/2+1 \le 1+\varepsilon/2$ & --- & no \\
\midrule
\multicolumn{4}{@{}l}{\emph{Episode $k \ge 2$} --- entered with $\abs{x} \le 1+\varepsilon$, \ $L \le \varepsilon/(1{+}\Delta)$.}\\
$t_k{+}1$ & $p_k+2$ \ (probe) & re-identifies: $L_k \le 2/(p_k{-}2)$ & yes \\
$t_k{+}2$ & $\tfrac{L}{2}(p_k{+}2)+1$ \ (echo) & --- & yes \\
$t_k{+}3$ & $\tfrac{1}{p_k-2}\big(\tfrac{L}{2}(p_k{+}2)+1\big)+1 \le 1+\varepsilon$ & --- & no \\
between   & $1+\varepsilon/2$ \ (invariant \eqref{eq:invariant}) & --- & no \\
\bottomrule
\end{tabular}
\caption{Per-step bookkeeping of the probe episodes of Appendix~\ref{app:regret}(i), valid for \emph{every} admissible $(a,w)$; a mistake is $\abs{x_t} > 1+\varepsilon$. The probe's own transition re-identifies $a$ whatever the adversary plays --- the disturbance shifts the window's center, never its width --- so each episode contributes at most $3$ mistakes before the loop re-enters the mistake-free regime of Figure~\ref{fig:probe}(a).}
\label{tab:episodes}
\end{table}

\begin{figure}[tbp]
\centering
\begin{tikzpicture}

\begin{scope}[y=0.5cm]
  \node[figlab, anchor=south west] at (-0.35,6.05)
       {\pnl{a} the probing law: bounded regret, unbounded peak};
  \draw[axisline] (-0.3,0) -- (10.1,0) node[fignum, above left=-1pt and -2pt] {$t$};
  \draw[axisline] (-0.25,-0.15) -- (-0.25,5.9) node[fignum, right=1pt] {$\abs{x_t}$};

  \fill[cPost!18] (-0.25,0) rectangle (9.9,0.62);
  \node[figtiny, cPost, anchor=north west, inner sep=2pt] at (-0.3,-0.75)
       {shaded: $\abs{x_t}\le 1+\varepsilon$ --- no mistakes between episodes};

  \foreach \X/\P/\E/\lb in {1.7/2.3/1.15/{$t_1$}, 4.7/3.7/1.5/{$t_2$},
                            7.7/5.3/1.9/{$t_3$}}{
     \draw[cAdv!55, line width=.9pt] (\X,0) -- (\X,\P);
     \node[trajdot] at (\X,\P) {};
     \draw[cAdv!55, line width=.9pt] (\X+0.42,0) -- (\X+0.42,\E);
     \node[trajdot, minimum size=2.8pt] at (\X+0.42,\E) {};
     \node[fignum, anchor=north, inner sep=2.5pt] at (\X,0) {\lb};}
  \node[figtiny, cAdv, anchor=south, inner sep=2.5pt] at (1.7,2.35) {$p_1$};
  \node[figtiny, cAdv, anchor=south, inner sep=2.5pt] at (4.7,3.75) {$p_2$};
  \node[figtiny, cAdv, anchor=south, inner sep=2.5pt] at (7.7,5.35) {$p_3$};
  \node[figtiny, cAdv, anchor=west, inner sep=3pt] at (8.35,5.0)
       {$p_k\uparrow\infty$};
  \node[figtiny, black!62, anchor=west, inner sep=2pt] at (2.42,1.95) {echo};
  \draw[gray!70, line width=.4pt] (2.38,1.9) -- (2.15,1.3);

  \foreach \X in {3.2, 6.2}{
     \draw[white, line width=2.2pt] (\X-0.06,-0.13) -- (\X+0.06,0.13);
     \draw[gray!70, line width=.4pt] (\X-0.14,-0.16) -- (\X-0.02,0.16);
     \draw[gray!70, line width=.4pt] (\X+0.02,-0.16) -- (\X+0.14,0.16);}
\end{scope}

\begin{scope}[yshift=-3.6cm, x=1cm, y=1cm]
  \node[figlab, anchor=south west] at (-0.35,2.25)
       {\pnl{b} inside one episode: why only $3$ mistakes};
  \draw[axisline] (-0.3,0.55) -- (5.5,0.55);
  \foreach \i/\h/\lb/\tk in {0/0.18/{}/{$t_k$}, 1/1.25/{probe}/{$+1$},
                             2/0.72/{echo}/{$+2$}, 3/0.20/{}/{$+3$},
                             4/0.18/{}/{$+4$}}{
     \fill[cAdv!30] (\i*1.05+0.1,0.55) rectangle (\i*1.05+0.55,0.55+\h);
     \draw[cAdv, line width=.5pt] (\i*1.05+0.1,0.55) rectangle (\i*1.05+0.55,0.55+\h);
     \node[figtiny, cAdv, anchor=south, inner sep=1.5pt]
          at (\i*1.05+0.325,0.58+\h) {\lb};
     \node[fignum, anchor=north, inner sep=2.5pt] at (\i*1.05+0.325,0.52) {\tk};}
  \draw[gray!55, dashed, line width=.4pt] (-0.1,0.73) -- (5.4,0.73);
  \node[figtiny, black!62, anchor=east, inner sep=2pt] at (-0.12,0.73)
       {$1+\varepsilon$};
  \node[figtiny, anchor=west, align=left, inner sep=2pt] at (3.55,1.55)
       {$\le3$ mistakes,\\ then back below $1+\varepsilon$};

  \fill[priorband] (0.15,-0.30) rectangle (2.0,-0.06);
  \draw[priorband, fill=none] (0.15,-0.30) rectangle (2.0,-0.06);
  \node[figtiny, anchor=east, inner sep=3pt] at (0.15,-0.18) {$L_{k-1}$};
  \fill[postband] (2.55,-0.30) rectangle (2.72,-0.06);
  \draw[postband, fill=none] (2.55,-0.30) rectangle (2.72,-0.06);
  \node[figtiny, anchor=west, inner sep=3pt, cPost] at (2.78,-0.18)
       {$L_k\le 2/(p_k-2)$};
  \node[figtiny, anchor=west, inner sep=3pt, black!62] at (0.15,-0.62)
       {the probe's own regressor identifies $a$ --- whatever the adversary plays};
  \draw[-{Latex[length=1.6mm,width=1.1mm]}, gray!75, line width=.5pt]
       (2.08,-0.18) -- (2.48,-0.18);
\end{scope}

\begin{scope}[yshift=-7.45cm]
  \node[figlab, anchor=south west] at (-0.35,2.30)
       {\pnl{c} the same two laws on three scoreboards \ ($\checkmark$ = better)};
  \node[figtiny, cPost, anchor=center] at (5.6,1.72) {optimal law};
  \node[figtiny, cAdv,  anchor=center] at (9.2,1.72) {probing law};
  \draw[gray!50, line width=.4pt] (4.3,1.50) -- (10.5,1.50);
  \fill[cAdv!8, rounded corners=2pt] (1.45,-0.53) rectangle (10.5,0.03);
  \node[figtiny, anchor=east, inner sep=3pt] at (3.6,1.15) {worst-case peak};
  \node[figtiny, anchor=center] at (5.6,1.15) {$\checkmark\; 1+\Delta$};
  \node[figtiny, black!45, anchor=center] at (9.2,1.15) {$\infty$};
  \node[figtiny, anchor=east, inner sep=3pt] at (3.6,0.45) {mistake regret};
  \node[figtiny, anchor=center] at (5.6,0.45) {$\checkmark\; O(1)$};
  \node[figtiny, black!45, anchor=center] at (9.2,0.45) {$\Theta(\log T)$};
  \node[figtiny, anchor=east, inner sep=3pt] at (3.6,-0.25) {state-cost regret};
  \node[figtiny, black!45, anchor=center] at (5.6,-0.25) {$\Theta(T)$};
  \node[figtiny, anchor=center] at (9.2,-0.25) {$\checkmark\; O(\sqrt{T})$};
  \node[figtiny, cAdv, anchor=north, inner sep=4pt] at (6.4,-0.58)
       {same trajectories, opposite ranking};
\end{scope}

\end{tikzpicture}
\caption{Appendix~\ref{app:regret}(i) and (iii): a law whose regret is small
under every accounting and whose worst-case peak is infinite.
\pnl{a} The optimal law with isolated probes $p_k\uparrow\infty$ superimposed at
a sparse schedule. Between episodes the invariant \eqref{eq:invariant} holds the
state below $1+\varepsilon$, so no mistakes accrue there. The regret is
carried entirely by the episodes, and there are only $\#\{k:t_k\le T\}$ of them.
\pnl{b} Inside an episode the bookkeeping is short because the probe is
self-correcting: the probe state is enormous, so the transition \emph{out} of it
measures $a$ through a huge regressor and confines the parameter to a window of
width $\le 2/(p_k-2)$, whatever the adversary does, since the disturbance can
shift that window's center but not its width. Two large states (probe and echo),
then the loop is back under $1+\varepsilon$.
\pnl{c} The consequence, and the reason Proposition~\ref{prop:regret} is stated
in both directions. Columns are fixed per law; the check marks the better law on
each scoreboard. The peak and mistake criteria prefer the optimal law; on the
third scoreboard the check jumps columns --- the integrated criterion prefers
the law with the \emph{infinite peak}. Nothing about the trajectories changed
between the rows --- only how they are added up.}
\label{fig:probe}
\end{figure}
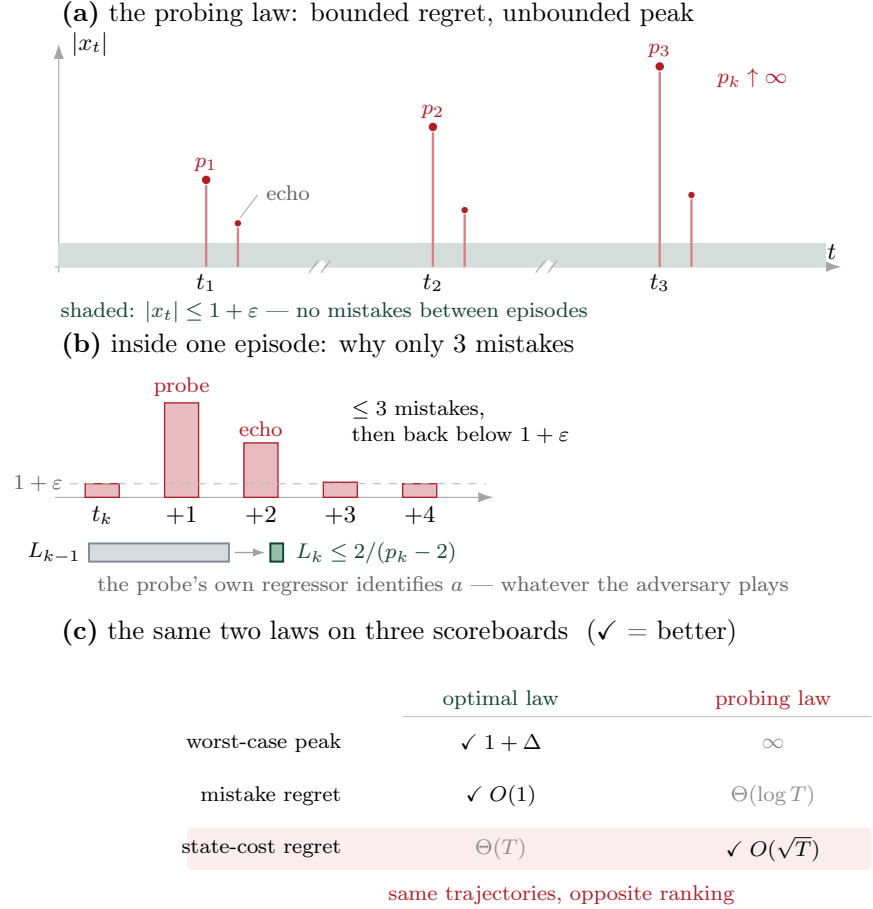

\paragraph{(ii) Optimal peak, linear regret.}
The attack is pictured in Figure~\ref{fig:stealth} and replayed, for both variants, in Table~\ref{tab:stealth}. Let $\Delta \ge 1$, $a = 1$, $w_0 = \tfrac{1}{\Delta+1}$, $w_t = 0$ for $t \ge 1$, and run the optimal law. Then $x_1 = w_0 = \tfrac{1}{\Delta+1}$ and, inductively, $\hat a_t = 0$ and $x_{t+1} = a x_t = x_t$: the state holds at $\tfrac{1}{\Delta+1}$ forever. The induced data window at each step has width $2/x_t = 2\Delta + 2$ and is centered at the true $a = 1$, hence equals $[1 - (\Delta+1),\, 1 + (\Delta+1)] \supseteq [-\Delta, \Delta]$: the consistent set never shrinks, $\hat a_t \equiv 0$, and the loop reproduces itself --- \emph{zero information leaks at any time}. The controller's cumulative state cost is $\tfrac{1}{\Delta+1}\,T$ ($t = 1, \dots, T$), while the clairvoyant (deadbeat at $a = 1$) pays only the one-time disturbance $\tfrac{1}{\Delta+1}$; the regret is exactly $\tfrac{1}{\Delta+1}(T-1)$. (Numerically: at $\Delta = 5$, $T = 2000$, measured regret $333.17 = \tfrac{1999}{6}$, slope exactly $\tfrac{1}{\Delta+1} = \tfrac16$.)

Two robustness properties, cited from Proposition~\ref{prop:regret}(ii) and Remark~\ref{rem:accounting}. \emph{Cost-insensitivity within the integrated family:} along the attack $u_t \equiv 0$, so any per-step cost charging the held state is paid linearly --- under $x_t^2$ and $x_t^2 + u_t^2$ the slope is $\bar x^2 = \tfrac{1}{(\Delta+1)^2}$, verified to six digits at $\Delta \in \{1, 1.5, 5, 50\}$, with $\max_t\abs{u_t}$ at machine zero and the consistent set equal to the full prior throughout. \emph{No restriction on $\Delta$:} for $\Delta < 1$, where $a = 1$ is inadmissible, take $a = \Delta$, $w_0 = \bar x := \tfrac{1}{\Delta+1}$, and $w_t \equiv (1 - \Delta)\bar x$ for $t \ge 1$. The state again holds at $\bar x$, the induced window $[1 - 1/\bar x,\, 1 + 1/\bar x] = [-\Delta,\, 2 + \Delta] \supseteq [-\Delta, \Delta]$ retains the full prior (its lower edge touching $-\Delta$ exactly), and the per-step excess over the clairvoyant is $\bar x - (1 - \Delta)\bar x = \tfrac{\Delta}{\Delta+1}$: linear regret with the stated slope, verified at $\Delta \in \{0.25, 0.5, 0.75\}$ to five digits.

\begin{table}[t]
\centering
\footnotesize
\setlength{\tabcolsep}{4.5pt}
\begin{tabular}{@{}l c c c c c@{}}
\toprule
$t$ & $x_t$ & $w_t$ & window $W_t$ & $A_{t+1}$ & excess $\abs{x_t} - \abs{x^*_t}$ \\
\midrule
\multicolumn{6}{@{}l}{\emph{$\Delta \ge 1$:} \ $a = 1$, \ $\bar x := \tfrac{1}{\Delta+1}$.}\\
$0$      & $0$      & $\bar x$ & $\R$ \ (zero regressor) & $[-\Delta,\Delta]$ & $0$ \\
$1$      & $\bar x$ & $0$      & $[-\Delta,\; \Delta+2]$ & $[-\Delta,\Delta]$ & $0$ \\
$t \ge 2$ & $\bar x$ & $0$      & $[-\Delta,\; \Delta+2]$ & $[-\Delta,\Delta]$ & $\bar x$ \\
\midrule
\multicolumn{6}{@{}l}{\emph{$\Delta < 1$:} \ $a = \Delta$, \ $\bar x := \tfrac{1}{\Delta+1}$.}\\
$0$      & $0$      & $\bar x$            & $\R$ \ (zero regressor) & $[-\Delta,\Delta]$ & $0$ \\
$1$      & $\bar x$ & $(1{-}\Delta)\bar x$ & $[-\Delta,\; 2+\Delta]$ & $[-\Delta,\Delta]$ & $0$ \\
$t \ge 2$ & $\bar x$ & $(1{-}\Delta)\bar x$ & $[-\Delta,\; 2+\Delta]$ & $[-\Delta,\Delta]$ & $\tfrac{\Delta}{\Delta+1}$ \\
\bottomrule
\end{tabular}
\caption{The stealth pairs of Appendix~\ref{app:regret}(ii) against the optimal law. Throughout, $\hat a_t \equiv 0$ and $u_t \equiv 0$: every induced window $W_t$ contains the full prior, so $A_{t+1} = A_t = [-\Delta, \Delta]$ and the loop reproduces itself --- zero information leaks at any time, and each row for $t \ge 1$ repeats forever. Excess is measured against the clairvoyant $x^*_{t+1} = w_t$ (deadbeat at the true $a$); summing the last column gives the stated linear regret.}
\label{tab:stealth}
\end{table}

\begin{figure}[tbp]
\centering
\begin{tikzpicture}

\begin{scope}[x=0.40cm, y=1cm]
  \node[figlab, anchor=south west] at (-6.6,3.05)
       {\pnl{a} why nothing is ever learned ($\Delta=5$)};
  \draw[axisline] (-6.6,-0.55) -- (8.6,-0.55)
        node[fignum, above left=-1pt and -2pt] {$\alpha$};
  \foreach \v/\lb in {-5/{$-\Delta$}, 0/{$0$}, 5/{$\Delta$}}{
     \draw[tick] (\v,-0.67) -- (\v,-0.43);
     \node[fignum, anchor=north, inner sep=2pt] at (\v,-0.67) {\lb};}

  \fill[priorband] (-5,2.07) rectangle (5,2.33);
  \draw[priorband, fill=none] (-5,2.07) rectangle (5,2.33);
  \node[figtiny, anchor=west, inner sep=0pt] at (7.3,2.20) {$A_t$};

  \fill[winband] (-5,1.22) rectangle (7,1.48);
  \draw[winband, fill=none] (-5,1.22) rectangle (7,1.48);
  \node[figtiny, anchor=west, inner sep=0pt, cWin] at (7.3,1.35) {$W_t$};
  \node[advdot] at (1,1.35) {};
  \node[fignum, cAdv, anchor=south, inner sep=2.5pt] at (1,1.50) {$a=1$};
  \draw[meas] (-5,1.08) -- (7,1.08)
       node[measlab, midway, below=7pt] {$2/\bar x = 2(\Delta+1)$};

  \fill[postband] (-5,0.12) rectangle (5,0.38);
  \draw[postband, fill=none] (-5,0.12) rectangle (5,0.38);
  \node[figtiny, anchor=west, inner sep=0pt, cPost] at (7.3,0.25)
       {$A_{t+1}=A_t$};

  \draw[cAdv, dash pattern=on 2pt off 1.5pt, line width=.5pt] (-5,0.00) -- (-5,2.62);
  \node[callout, cAdv, anchor=south west, align=left] at (-4.8,2.50)
       {lower edge lands exactly on $-\Delta$};
\end{scope}

\begin{scope}[xshift=7.75cm, yshift=-0.60cm]
\begin{axis}[appfig, width=5.0cm, height=3.6cm,
    xlabel={horizon $T$}, ylabel={cumulative $\sum_t\abs{x_t}$},
    xmin=0, xmax=12, ymin=0, ymax=2.35,
    xtick={0,4,8,12}, ytick={0,1,2},
    title={\pnl{b} the resulting regret}]
  \addplot[draw=none, fill=cAdv!12, forget plot] coordinates
      {(1,0.1667) (12,0.1667) (12,2.0) (1,0.1667)} \closedcycle;
  \addplot[cAdv, domain=1:12, samples=2] {x/6};
  \addplot[cPost, domain=1:12, samples=2] {1/6};
  \node[figtiny, cAdv, anchor=north east] at (axis cs:11.4,1.62)
       {optimal law};
  \node[figtiny, cPost, anchor=south west] at (axis cs:5.2,0.22)
       {clairvoyant};
  \node[figtiny, cAdv!85, anchor=east, align=right] at (axis cs:10.6,0.95)
       {$R_T=\frac{T-1}{\Delta+1}$};
\end{axis}
\end{scope}

\end{tikzpicture}
\caption{The stealth attack of Appendix~\ref{app:regret}(ii), at $\Delta=5$:
$a=1$, $w_0=\bar x:=\tfrac{1}{\Delta+1}$, $w_t=0$ thereafter.
\pnl{a} The adversary parks the state at $\bar x$ and leaves it there. Because
the regressor is $\bar x$, the induced data window has width $2/\bar x =
2(\Delta+1)$ and is centered at the true $a=1$, so it contains the entire
prior, with its lower edge falling exactly on $-\Delta$. The intersection removes
nothing, the consistent set never shrinks, the midpoint stays at $0$, the input
stays at $0$, and the loop reproduces itself forever: an attack that leaks
\emph{no} information at any time. This is the exact complement of the spike of
Remark~\ref{rem:anatomy}: there, magnitude is bought at the price of information;
here, no magnitude is taken and no information is given.
\pnl{b} The optimal law therefore holds $\bar x$ forever while the clairvoyant,
knowing $a$, pays only the one-time disturbance. Under the cumulative state cost
the minimax-optimal policy accrues regret $\tfrac{T-1}{\Delta+1}$, linear in $T$
--- the price of refusing to gamble on dynamics the data have not pinned down.}
\label{fig:stealth}
\end{figure}

\paragraph{(iii) Both directions under the single cost $c_t = \abs{x_t}$ (Remark~\ref{rem:accounting}).}
Take the geometric schedule $t_k = 2^{k-1} t_1$ with probe sizes $p_k = c_0\sqrt{t_k}$, where $p_1 = c_0\sqrt{t_1}$ is sized as in (i), and measure regret against the same clairvoyant (deadbeat at the true $a$: $x^*_{t+1} = w_t$). Fix any admissible $(a,w)$. The per-step excess is $\abs{x_{t+1}} - \abs{x^*_{t+1}} = \abs{(a - \hat a_t)x_t + \pi_t + w_t} - \abs{w_t} \le \abs{a - \hat a_t}\abs{x_t} + \abs{\pi_t}$, where $\pi_t$ is the probe offset (zero off the schedule), and the sum splits into four parts. \emph{Pre-probe:} at most $t_1(1 + \Delta) = O(1)$. \emph{Spikes:} $\sum_{k : t_k \le T} p_k \le \tfrac{c_0}{1 - 2^{-1/2}}\sqrt{T}\,(1 + o(1))$. \emph{Echoes:} the step after probe $k$ costs at most $\tfrac{L_{k-1}}{2}(p_k + O(1)) + 1 = O(p_k/p_{k-1}) = O(1)$, since the width available during the spike step is the pre-probe width $L_{k-1} \le 2/(p_{k-1} - 2)$. \emph{Drift:} from the probe-$k$ transition onward, $\abs{a - \hat a_t} \le L_k/2 \le 1/(p_k - 2)$ --- the disturbance shifts only the center of the probe's data window, never its width $2/\abs{x_{t_k+1}}$, so this holds whatever $(a,w)$ --- while Proposition~\ref{prop:invariant} keeps $\abs{x_t} = O(1)$ between probes: segment $k$ contributes $O\big((t_{k+1} - t_k)/p_k\big) = O(\sqrt{t_k})$, again summing geometrically. Hence $R_T \le C'(\Delta, c_0, t_1)\,\sqrt{T}$ for \emph{every} admissible $(a,w)$, while the peak is infinite along $(a,w) = (0,0)$. Numerically, at $\Delta = 5$, $c_0 = 10$, $T = 2^{21}$, against a battery of $80$ adversaries (fixed-parameter greedy, mismatch-aligned, probe-aligned, window-poisoning, stealth, and random families; structured attacks replayed from certified admissible pairs as in Section~\ref{sec:numerics}), the supremum of $R_T/\sqrt{T}$ over all adversaries and checkpoints is $34.2$, matching the predicted spike-sum constant $c_0/(1 - 2^{-1/2}) \approx 34.1$, with log--log regret slope $0.505$ under $(a,w) = (0,0)$; a squaring schedule $t_{k+1} = t_k^2$ with $p_k = t_k^{2/3}$ measures $\sup_T R_T/T^{2/3} = 1.3$, the rate the same decomposition predicts. The geometric schedule is the better tuning. Under the threshold cost the same two laws rank in the \emph{opposite} order: over the same battery the plain optimal law's total mistakes never exceeded $3$ ($\Delta \in \{1, 5, 25\}$, $T = 10^5$, $\varepsilon \in \{0.1, 0.25, 0.5\}$): bounded mistake-regret, against $\Theta(\#\text{probes})$ for the probing law.

\end{document}